\documentclass[11pt]{article}

\usepackage[nolists,nomarkers,tablesfirst]{endfloat}

\usepackage{amsfonts,amsmath,mathtools,mleftright}
\usepackage{graphicx,subfig,multirow}
\usepackage{bbm,bm} 
\usepackage[colorlinks=true,citecolor=blue,urlcolor=blue,linkcolor=blue]{hyperref}
\usepackage{chicago} 
\usepackage[margin=1in]{geometry}

\newcommand{\e}{{\rm e}}
\newcommand{\E}{{\mathbb E}}
\newcommand{\F}{{\mathbb F}}

\newcommand{\Pa}{{\mathbb P}}
\newcommand{\Q}{{\mathbb Q}}

\newcommand{\R}{{\mathbb R}}

\newcommand{\N}{{\mathbb N}}

\newcommand{\Acal}{{\mathcal A}}

\newcommand{\Ecal}{{\mathcal E}}
\newcommand{\Fcal}{{\mathcal F}}
\newcommand{\Gcal}{{\mathcal G}}
\newcommand{\Hcal}{{\mathcal H}}

\newcommand{\Lcal}{{\mathcal L}}

\newcommand{\Xcal}{{\mathcal X}}
\newcommand{\Ycal}{{\mathcal Y}}

\newcommand{\cExp}[3][]{\E^{#1} \mleft[ {#2} \, \middle| \, {#3} \mright]}
\newcommand{\Ind}[1]{\mathbbm{1}_{\left\{ #1 \right\} }}
\newcommand{\im}{{\rm i}}

\DeclareMathOperator{\diag}{diag}
\DeclareMathOperator{\Id}{{\rm Id}}

\newtheorem{proposition}{Proposition}[section]
\newtheorem{lemma}[proposition]{Lemma}
\newtheorem{theorem}[proposition]{Theorem}

\newtheorem{corollary}[proposition]{Corollary}
\newtheorem{remark}[proposition]{Remark}

\newtheorem{example}[proposition]{Example}

\begin{document}
\title{Linear credit risk models
\footnote{The authors would like to thank for useful comments Agostino Capponi, David Lando, Martin Larsson, Jongsub Lee, Andrea Pallavicini, Sander Willems, and two anonymous referees, as well as participants from the 2015 AMaMeF and Swissquote conference in Lausanne, the 2016 Bachelier world congress in New-York, the 2016 EFA annual meeting in Oslo, the 2016 AFFI Paris December meeting, and the 2017 CIB workshop on Dynamical Models in Finance.
The research leading to these results has received funding from the European Research Council under the European Union's Seventh Framework Programme (FP/2007-2013) / ERC Grant Agreement n.~307465-POLYTE.}
}
\author{Damien Ackerer \footnote{Swissquote Bank. Email: damien.ackerer@swissquote.ch} \and Damir Filipovi\'c \footnote{EPFL and Swiss Finance Institute. Email: damir.filipovic@epfl.ch}}
\date{July 6, 2019}

\maketitle

\begin{center} 
\large forthcoming in \textit{Finance and Stochastics} 
\end{center} 

\bigskip 

\begin{abstract} 
We introduce a novel class of credit risk models in which the drift of the survival process of a firm is a linear function of the factors. 
The prices of defaultable bonds and credit default swaps (CDS) are linear-rational in the factors. 
The price of a CDS option can be uniformly approximated by polynomials in the factors.
Multi-name models can produce simultaneous defaults, generate positively as well as negatively correlated default intensities, and accommodate stochastic interest rates.
A calibration study illustrates the versatility of these models by fitting CDS spread time series. A numerical analysis validates the efficiency of the option price approximation method.
\end{abstract}

\smallskip
\noindent \textbf{Keywords:} credit default swap, credit derivatives, credit risk, polynomial model, survival process\\

\noindent \textbf{MSC (2010):} 91B25, 91B70, 91G20, 91G40, 91G60\\

\smallskip
\noindent \textbf{JEL Classification:} C51, G12, G13 

\newpage

\section{Introduction}\label{sec:intro}

We introduce a novel class of flexible and tractable reduced form models for the term structure of credit risk, the linear credit risk models. We directly specify the survival process of a firm, that is, its conditional survival probability given the economic background information. Specifically, we assume a multivariate factor process with a linear drift and let the drift of the survival process be linear in the factors. Prices of defaultable bonds and credit default swaps (CDS) are given in closed form by linear-rational functions in the factors. 
By linearity, the same result holds for the prices of CDSs on indices (CDISs).
The implied default intensity is a linear-rational function of the factors. 
In contrast, the price of a CDS in an affine default intensity model is a sum of exponential-affine functions in the factors process and whose coefficients are given by the solutions of nonlinear ordinary differential equations that are not in closed form, in general.
In addition, the linear credit risk models offer new tractable features such as a multi-name model with negatively correlated default intensity.

Within the linear framework we define the linear hypercube (LHC) model which is a single-name model.
The factor process is diffusive with quadratic diffusion function so that it takes values in a hypercube whose edges' length is given by the survival process.
The quadratic diffusion function is concave and bi-monotonic.
This feature allows factors to virtually jump between low and high values.
This facilitates the persistence and likelihood of term structure shifts.
The factors' volatility parameters do not enter the bond and CDS pricing formulas, yet they impact the volatility of CDS spreads and thus affect CDS option prices. 
This may facilitate the joint calibration of credit spread and option price time series. 
We discuss in detail the one-factor LHC model and compare it with the one-factor affine default intensity model. We provide an identifiable canonical representation and the market price of risk specifications that preserve the linear drift of the factors.

We present a price approximation methodology for European style options on credit risky underlyings that exploits the compactness of the state space and the closed form of the conditional moments of the factor process. First, by the Stone--Weierstrass theorem, any continuous payoff function on the compact state space can be approximated by a polynomial to any given level of accuracy. Second, the conditional expectation of any polynomial in the factors is a polynomial in the prevailing factor values. In consequence, the price of a CDS option can be uniformly approximated by polynomials in the factors. This method also applies to the computation of credit valuation adjustments.

We build multi-name models by letting the survival processes be linear and polynomial combinations of independent LHC models.
Bonds and CDSs prices are still linear-rational but with respect to an extended factor representation.
These direct extensions can easily accommodate the inclusion of new factors and new firms.
Stochastic short rate models with a similar specification as the survival processes can be introduced while preserving the setup tractability.
Simultaneous defaults can be generated either by introducing a common jump process in the survival processes or a stochastic clock.

We perform an empirical and numerical analysis of the LHC model. Assuming a parsimonious cascading drift structure, we fit two-factor and three-factor LHC models to the ten-year long time series of weekly CDS spreads on an investment grade and a high yield firm. The three-factor model is able to capture the complex term structure dynamics remarkably well and performs significantly better than the two-factor model. 
We illustrate the numerical efficiency of the option pricing method by approximating the prices of CDS options with different moneyness. 
Polynomials of relatively low orders are sufficient to obtain accurate approximations for in-the-money options. Out-of-the money options typically require a higher order. 
We derive the pricing formulas for CDIS options and tranches on a homogeneous portfolio to illustrate that their prices can also be approximated with similar techniques.
In general, the pricing of CDIS options and tranches requires manipulating multivariate polynomial bases of possibly large dimensions.
In practice, computationally efficient multi-name credit derivative pricing necessitates the use of special algorithms which are not studied in this paper.

\smallskip

We now review some of the related literature.
Our approach follows a standard doubly stochastic construction of default times as described in \cite{elliott2000models} or \cite{Bielecki2002credit}.
The early contributions by \cite{Lando1998} and \cite{Duffie1999} already make use of affine factor processes. 
In contrast, the factor process in the LHC model is a strictly non-affine polynomial diffusion, whose general properties are studied in \cite{filipovic2016polynomial}.
The stochastic volatility models developed in \cite{hull1987pricing} and \cite{ackerer2016jacobi} are two other examples of non-affine polynomial models.
Factors in the LHC models have a compact support and can exhibit jump-like dynamics similar to the multivariate Jacobi process introduced by~\cite{gourieroux2006multivariate}.
Our approach has some similarities with the linearity generating process by \cite{Gabaix2007} and the linear-rational models by \cite{filipovic2017linear}.
These models also exploit the tractability of factor processes with linear drift but focus on the pricing of non-defaultable assets. 
To our knowledge, we are the first to model directly the survival process of a firm with linear drift characteristics.

Options on CDS contracts are complex derivatives and intricate to price.
The pricing and hedging of CDS options in a generic hazard process framework is discussed in \cite{bielecki2006hedging} and \cite{bielecki2008pricing}, and specialised to the square-root diffusion factor process in \cite{bielecki2011hedging}.
More recently \cite{brigo2010exact} developed a semi-analytical expression for CDS option prices in the context of a shifted square-root jump-diffusion default intensity model that was introduced in \cite{brigo2005credit}.
Another strand of the literature has focused on developing market models in the spirit of LIBOR market models. 
We refer the interested reader to \cite{Schonbucher2000}, \cite{hull2003valuation}, \cite{Schonbucher2004}, \cite{Jamshidian2004}, and \cite{Brigo2005}. 
Black-Scholes like formulas are then obtained for the prices of CDS options by assuming, for example, that the underlying CDS spread follows a geometric Brownian motion under the survival measure.
Although offering more tractability, this approach makes it difficult, if not impossible, to consistently price multiple instruments exposed to the same source of credit risk.
\cite{di2009dynamic} introduced a framework where they obtained closed form expressions similar to ours for CDS prices, but under the assumption that the firm default intensity is driven by a continuous-time finite-state irreducible Markov chain.

Another important approach to default risk modeling is the use of subordinators to model the cumulative hazard process.
It has been in particular shown that time-inhomogeneous models can reproduce well CDIS tranche prices.
For more details on these models we refer to \cite{kokholm2010sato}, \cite{sun2017marshall}, and references therein.

The simulation-based work by~\cite{peng2008connecting} shows that a hazard-rate model with systemic and idiosyncratic risk factors can fit both CDS and CDIS tranches, and therefore confirms that a bottom-up model with common risk factors can yield an accurate and fully consistent risk-management framework.
A tractable alternative to price multi-name credit derivatives is to model the dependence between defaults with a copula function, as for example in~\cite{li2000default}, \cite{laurent2005basket}, and \cite{ackerer2016dependent}. 
However these models are by construction static, require repeated calibration, and in general become intractable when combined with stochastic survival processes as in~\cite{schonbucher2001copula}.

The idea of approximating option prices by power series can be traced back to \cite{jarrow1982approximate}.
However, most of the previous literature has focussed on approximating the transition density function of the underlying process, see for example \cite{corrado1996skewness} and \cite{filipovic2013density}.
In contrast, we approximate directly the payoff function by a polynomial.

\smallskip

The remainder of the paper is structured as follows.
Section~\ref{sec:LF} presents the linear credit risk framework along with generic pricing formulas.
Section~\ref{sec:LHCM} describes the single-name LHC model.
The numerical and empirical analysis of the LHC model is in Section~\ref{sec:app}.  
Multi-name models as well as models with stochastic interest rates are discussed in Section~\ref{sec:extensions}.
Section~\ref{sec:ccl} concludes.
The proofs are collected in the Appendix, as well as some additional results on market price of risk specifications that preserve the linear drift of the factors, and on the two-dimensional Chebyshev interpolation.

\section{The linear framework}\label{sec:LF}

We introduce the linear credit risk model framework and derive closed form expressions for defaultable bond prices and credit default swap spreads. We also discuss the pricing of credit index tranches, credit default swap options, and credit valuation adjustments.

\subsection{Survival process specification} \label{sec:survival}

We fix a probability space $(\Omega, \Fcal, \Q)$ equipped with a right-continuous filtration ${\F=(\Fcal_t)_{t\ge0}}$ representing the economic background information, and where $\Q$ is the risk-neutral pricing measure. We consider $N$ firms and let $S^i$ be the survival process of firm $i$. 
This is a right-continuous $\F$-adapted and non-increasing positive process with ${S_0^i = 1}$.
Let $U^1,\dots,U^N$ be mutually independent standard uniform random variables that are independent from $\Fcal_\infty$.
For each firm $i$, we define the random default time
\begin{equation*}
\tau_i = \inf\{t\geq 0 : S^i_t \leq U_i \},
\end{equation*}
which is infinity if the set is empty.
Let $(\Hcal^i_t)_{t\ge0}$ be the filtration generated by the indicator process which is one as long as firm $i$ has not defaulted by time $t$ and zero afterwards, $H^i_t = \Ind{\tau_i>t}$ for $t\ge0$.
The default time $\tau_i$ is a stopping time in the enlarged filtration $(\Fcal_t \vee \Hcal^i_t)_{t\ge0} $.
It is $\F$-doubly stochastic in the sense that
\begin{equation*}
\Q[\,\tau_i>t \mid \Fcal_\infty\,] = \Q[\, S^i_t > U_i \mid \Fcal_\infty\,] = S^i_t.
\end{equation*}
The filtration $(\Gcal_t)_{t\ge0}=(\Fcal_t \vee \Hcal^1_t \vee\dots\vee \Hcal^N_t)_{t\ge0}$ contains all the information about the occurrence of firm defaults, as well as the economic background information.
Henceforward we omit the index $i$ of the firm and refer to any of the $N$ firms as long as there is no ambiguity. 

In a linear credit risk model the survival process of a firm is defined by
\begin{equation}\label{eq:S}
S_t =  a^\top Y_{t}, \quad t\ge0,
\end{equation}
for some firm specific parameter $a\in\R^n_+$, and some common factor process $(Y,X)$ taking values in $\R^n_+\times\R^m$ with linear drift of the form
\begin{align}
dY_t &= (  c \, Y_t + \gamma \, X_t) dt + dM^Y_t \label{eq:dY}\\
dX_t &= (  b \, Y_t + \beta \, X_t)dt + dM^X_t \label{eq:dX}
\end{align}
for some $c\in\R^{n\times n}$, $b\in\R^{m\times n}$, $\gamma\in\R^{n\times m}$, $\beta\in\R^{m\times m}$, $m$-dimensional $\F$-martingale $M^X$, and $n$-dimensional $\F$-martingale $M^Y$.
The process $S$ being positive and non-increasing, we necessarily have that its martingale component $M^{S} = a^\top M^Y$ is of finite variation and thus purely discontinuous, see \cite[Lemma I.4.14]{jacod2013limit}, and that $- S_{t-} < \Delta M_t^{S} \le 0 $ for all $t\ge0$
because $\Delta S_t = \Delta M_t^{S}$.
This observation motivates the decomposition of the factor process into a component $X$ and a component $Y$ with finite variation.
Although we do not specify further the dynamics of the factor process at the moment, it is important to emphasize that additional conditions should be satisfied to ensure that $S$ is a valid survival process.
\begin{remark}
In practice we will consider a componentwise non-increasing process $Y$ with $Y_0={\bm 1}$. 
Survival processes can then easily be constructed by choosing any vector $a\in\R^n_+$ such that $a^\top{\bm 1} =1$.
\end{remark}

The linear drift of the process $(Y,X)$ implies that the $\Fcal_t$-conditional expectation of $(Y_u,X_u)$ is linear of the form
\begin{equation}\label{lineq}
\E \Big[ \begin{pmatrix} Y_{u} \\ X_{u} \end{pmatrix} \; \Big| \; \Fcal_t \Big] = \e^{ A(u-t)}\begin{pmatrix} Y_{t} \\ X_{t} \end{pmatrix},\quad t\le u,
\end{equation}
where the $(m+n)\times(m+n)$-matrix $A$ is defined by
\begin{equation} \label{eq:Amat}
A=
\begin{pmatrix}
c & \gamma \\
b & \beta
\end{pmatrix}.
\end{equation}

\begin{remark}\label{remhara}
If $S$ is absolutely continuous, so that $a^\top dM^Y_t=0$ for all $t\ge0$, the corresponding default intensity $\lambda$, which derives from the relation $S_t = \e^{-\int_0^t \lambda_s ds}$, is linear-rational in $(Y,X)$ of the form
\begin{equation*}
\lambda_t = -\frac{a^\top ( c \, Y_t + \gamma \, X_t)}{S_t} .
\end{equation*}
\end{remark}

In this framework, the default times are correlated because the survival processes are driven by common factors.
Simultaneous defaults are possible and may be caused by the martingale components of $Y$ that forces the survival processes to jump downward at the same time.
Additionally, and to the contrary of affine default intensity models, the linear credit risk framework allows for negative correlation between default intensities as illustrated by the following stylized example.

\begin{example}\label{ex:negcor}
Consider the factor process $(Y,X)$ taking values in $\R_+^2\times\R$ defined by
\begin{align*}
dY_t &= \frac\epsilon2 \left(\begin{pmatrix} -1 & 0  \\ 0 & -1 \end{pmatrix}Y_t + \begin{pmatrix} -1 \\ 1 \end{pmatrix}X_t \right)dt \\
dX_t & = -\kappa X_t dt + \sigma\sqrt{(\e^{-\epsilon t}-X_t)(\e^{-\epsilon t}+X_t)}dW_t
\end{align*}
for some $\kappa>\epsilon>0$, $\sigma>0$, $X_0\in[-1,1]$, and $\F$-adapted univariate Brownian motion $W_t$.
The process $X$ takes values in the interval $[-\e^{-\epsilon t},\e^{-\epsilon t}]$ at time $t$.
Let $N=2$ survival processes be defined by $S^1_t = Y_{1t}$ and $S^2_t = Y_{2t}$ for all $t\ge0$ so that the implied default intensities of the two firms are given by
\[
\lambda^1_t = \frac{\epsilon}{2}\left(1 + \frac{X_t}{Y_{1t}}\right)
\quad
\text{and}
\quad
\lambda^2_t = \frac{\epsilon}{2}\left(1-\frac{X_t}{Y_{2t}}\right), \quad t\ge0.
\]
This results in $d\langle\lambda^1,\lambda^2\rangle_t \le 0$ and  $d\langle\lambda^1,\lambda^2\rangle_t < 0$ with positive probability, and $\lambda_t^1,\lambda_t^2\le\epsilon$. 
Moreover, the default intensities, $\lambda^1$ and $\lambda^2$, both mean-revert towards $\epsilon/2$.
The proof of these statements is given in Appendix~\ref{sec:LCRMproofs}.
\end{example}

%
%
%

\subsection{Defaultable bonds}\label{sec:bonds}

We consider securities with notional amount equal to one and exposed to the credit risk of a reference firm.
We assume a constant risk-free interest rate equal to $r$ so that the time-$t$ price of the risk-free zero-coupon bond with maturity $t_M$ and notional amount one is given by $\e^{-r(t_M-t)}$.
The following result gives a closed form expression for the price of a defaultable bond with constant recovery rate at maturity.

\begin{proposition} \label{prop:bondMat}
The time-$t$ price of a defaultable zero-coupon bond with maturity $t_M$ and recovery $\delta\in[0,1]$ at maturity is
\begin{align*}
B_{\rm M}(t,t_M) &= \cExp{\e^{-r(t_M-t)} \big(\Ind{\tau>t_M} + \delta \Ind{\tau \le t_M} \big)}{\Gcal_t}\\
& = (1-\delta)B_{\rm Z}(t,t_M) + \Ind{\tau>t} \delta \, \e^{-r(t_M-t)}
\end{align*}
where
$B_{\rm Z}(t,t_M) = \e^{-r(t_M-t)}\E[\Ind{\tau>t_M} \mid \Gcal_t]$
denotes the time-$t$ price of a defaultable zero-coupon bond with maturity $t_M$ and zero recovery. It is of the form
\begin{equation}\label{eq:Zbond}
B_{\rm Z}(t,t_M) = \Ind{\tau>t} \frac{1}{a^\top Y_t} \, \psi_{\rm Z}(t,t_M)^\top \begin{pmatrix} Y_t \\ X_t \end{pmatrix}
\end{equation}
where the vector $\psi_{\rm Z}(t,t_M)\in\R^{n+m}$ is given by
\begin{equation*}
\psi_{\rm Z}(t,t_M)^\top = \e^{-r(t_M-t)} 
\begin{pmatrix}
a^\top & {\bf 0}^\top_m
\end{pmatrix} \e^{ A(t_M-t)}
\end{equation*}
where the $m$-dimensional vector ${\bf 0}_m$ contains only zeros.
\end{proposition}

The next result shows that the price of a defaultable bond paying a constant recovery rate at default can also be retrieved in closed form.

\begin{proposition} \label{prop:bondDef}
The time-$t$ price of a defaultable zero-coupon bond with maturity $t_M$ and recovery $\delta\in[0,1]$ at default is
\begin{align*}
B_{\rm D}(t,t_M) &= \cExp{\e^{-r(t_M-t)}\Ind{\tau>t_M} + \delta \e^{-r(\tau-t)} \Ind{t < \tau \le t_M}}{\Gcal_t} \\
&= B_{\rm Z}(t,t_M) + \delta \, C_{\rm D}(t,t_M),
\end{align*}
where $C_{\rm D}(t,t_M) = \E[\e^{-r(\tau-t)} \Ind{t<\tau \le t_M} | \Gcal_t]$
denotes the time-$t$ price of a contingent claim paying one at default if it occurs between dates $t$ and $t_M$.  It is of the form
\begin{equation}\label{eq:CFdef}
C_{\rm D}(t,t_M) = \Ind{\tau > t}\frac{1}{a^\top Y_t} \, 
 \psi_{\rm D}(t,t_M)^\top  \begin{pmatrix} Y_t \\ X_t \end{pmatrix}
\end{equation}
where the vector $\psi_{\rm D}(t,t_M)\in\R^{n+m}$ is given by
\begin{equation} \label{eq:psiD}
\psi_{\rm D}(t,t_M)^\top =  -a ^\top \begin{pmatrix} c & \gamma \end{pmatrix} 
 \int_t^{t_M}  \e^{A_*(s-t)}ds 
\end{equation}
where $A_*=A-r \Id$.
\end{proposition}

The price of a security whose only cash flow is proportional to the default time is given in the following corollary.
It is used to compute the expected accrued interests at default for some contingent securities such as CDS.

\begin{corollary}\label{cor:condBond}
The time-$t$ price of a contingent claim paying $\tau$ at default if it occurs between date $t$ and $t_M$ is of the form
\begin{align}\label{eq:CFdefs}
\begin{split}
C_{\rm D_*}(t,t_M) &= \cExp{\tau \e^{-r(\tau-t)} \Ind{\tau \le t_M}}{\Gcal_t} \\
&= \Ind{\tau>t} \frac{1}{a^\top Y_t} \, \psi_{\rm D_*}(t,t_M)^\top \begin{pmatrix} Y_t \\ X_t \end{pmatrix}
\end{split}
\end{align}
where the vector $\psi_{\rm D_*}(t,t_M)\in\R^{n+m}$ is given by
\begin{equation}\label{eq:psiDs}
\psi_{\rm D_*}(t,t_M)^\top = -a ^\top \begin{pmatrix} c & \gamma \end{pmatrix}  \;  \int_t^{t_M} s \, \e^{A_*(s-t)} ds.
\end{equation}
\end{corollary}
Note the presence of the factor $s$ in the integrand on the right hand side of \eqref{eq:psiDs}, which is absent in \eqref{eq:psiD}.

\begin{remark}
By setting $r=0$ in~\eqref{eq:CFdefs}, we obtain a closed form expression for $\E[\tau \Ind{\tau\le t_M}\mid\Gcal_t]$. 
This expression can be used to price a defaultable bond whose recovery value at maturity $t_M$ depends on the default time $\tau$ in a linear way,
\[
B_{\rm D_0}(t, t_0, t_M) = B_{\rm Z}(t, t_M) + \e^{-r(t_M-t)} \E \Big[ \big( \delta_0 \frac{\tau-t_0}{t_M-t_0} + \delta_1 \big) \Ind{\tau\le t_M} \; \Big| \; \Gcal_t \Big]
 \]
for some parameters $\delta_0,\delta_1\ge0$ such that $\delta_0 + \delta_1 \le 1$, and for some time $t_0\le t$.
\end{remark}
The following Lemma shows that pricing formulas~\eqref{eq:CFdef}--\eqref{eq:psiDs} can further simplify with an additional condition.

\begin{lemma} \label{lem:simplerFormulas}
Assume that the matrix $A_*$ is invertible then we have the following closed form expressions
\begin{align*}
\psi_{\rm D}(t,t_M)^\top &= 
-a ^\top \begin{pmatrix} c & \gamma \end{pmatrix}   A_*^{-1} \left( \e^{ {A_*}(t_M-t)} - \Id \right) \\
\psi_{\rm D_*}(t,t_M)^\top &=
-a ^\top \begin{pmatrix} c & \gamma \end{pmatrix}  \Big( (t_M-t)   A_*^{-1}\e^{A_*(t_M-t)}  \\
& \quad + A_*^{-1}(\Id t - A_*^{-1})(\e^{A_*(t_M-t)} - \Id) \Big)
\end{align*}
where $\Id$ is the $(n+m)$-dimensional identity matrix.
\end{lemma}

This is a remarkable result since the prices of contingent cash flows become closed form expressions composed of basic matrix operations and are thus easily computed.
Closed form formulas for defaultables securities render the linear framework appealing for large scale applications, for example with a large number of firms and contracts, in comparison to standard affine default intensity models that in general require the use of additional numerical methods.
For illustration, assume that the survival process $S_t$ is absolutely continuous so that it admits the default intensity $\lambda_t$ as in Remark~\ref{remhara}. Then $C_{\rm D}(t,t_M)$ can be rewritten as
\begin{align*}
C_{\rm D}(t,t_M) = \Ind{\tau>t}\int_t^{t_M}\e^{-r(u-t)}\,\cExp{\lambda_u\e^{-\int_t^u \lambda_sds}}{\Fcal_t}du.
\end{align*}
With affine default intensity models the expectation to be integrated requires solving Riccati equations, which have a closed form solution only when the default intensity is driven by a sum of independent univariate CIR processes.
Numerical methods such as finite difference are usually employed to compute the expectation with time-$u$ cash flow for $u\in[t,t_M]$.
The integral can then only be approximated by means of another numerical method such as quadrature, that necessitates solving the corresponding ordinary differential equations at many different points $u$.
For more details on affine default intensity models we refer to~\cite{duffie2003credit}, \cite{filipovic2009term}, and~\cite{lando2009credit}.

\subsection{Credit default swaps} \label{sec:CDS}

We derive closed form expressions for credit default swaps (CDS) on a single firm and multiple firms.
We conclude the section with a discussion of factors unspanned by bonds and CDS prices.

A \emph{single-name CDS} is an insurance contract that pays at default the realized loss on a reference bond -- the protection leg -- in exchange of periodic payments that will stop after default -- the premium leg.
We consider the following discrete tenor structure $t \le t_0 < t_1 < \dots < t_M$ and a contract offering default protection from date $t_0$ to date $t_M$.
When $t<t_0$ the contract is usually called a knock out forward CDS and generates cash flows only if the firm has not defaulted by time $t_0$.
We consider a CDS contract with notional amount equal to one.
The time-$t$ value of the premium leg with spread $k$ is given by $k \, V_{\rm prem}(t,t_0,t_M)$ where
\[ V_{\rm prem}(t,t_0,t_M) =   V_{\rm coup}(t,t_0,t_M)+ V_{\rm ai}(t,t_0,t_M)  \]
is the sum of the value of coupon payments before default
\[ V_{\rm coup}(t,t_0,t_M) = \sum_{j=1}^M \; \cExp{  \e^{-r(t_j-t)}(t_j - t_{j-1})\Ind{t_j<\tau}}{\Gcal_t}\]
and the value of the accrued coupon payment at the time of default
\[ V_{\rm ai}(t,t_0,t_M) = \sum_{j=1}^M \; \cExp{ \e^{-r(\tau-t)}(\tau - t_{j-1})\Ind{t_{j-1}<\tau\leq t_j}}{\Gcal_t}.\]
The time-$t$ value of the protection leg is
\[ V_{\rm prot}(t,t_0,t_M) = (1-\delta)\cExp{\e^{-r(\tau-t)}\Ind{t_0<\tau\leq t_M}}{\Gcal_t} ,\]
where $\delta\in[0,1]$ denotes the constant recovery rate at default.
This specification of payments is in line with the ISDA model, see~\cite{white2013pricing}.
The (forward) CDS spread ${\rm CDS}(t,t_0,t_M)$ is the spread $k$ that makes the premium leg and the protection leg equal in value at time $t$. That is,
\[ {\rm CDS}(t,t_0,t_M) = \frac{V_{\rm prot}(t,t_0,t_M)}{V_{\rm prem}(t,t_0,t_M)}.\]

\begin{proposition}\label{prop:cds}
The values of the protection and premium legs are given by
\begin{align}
V_{\rm prot}(t,t_0,t_M) & =  \Ind{\tau>t}\frac{1}{S_t} \psi_{\rm prot}(t,t_0,t_M)^\top \begin{pmatrix} Y_t \\ X_t \end{pmatrix} \label{eq:Vprot}\\
V_{\rm prem}(t,t_0,t_M) & =  \Ind{\tau>t}\frac{1}{S_t} \psi_{\rm prem}(t,t_0,t_M)^\top \begin{pmatrix} Y_t \\ X_t \end{pmatrix}\label{eq:Vprem}
\end{align}
where the vectors $\psi_{\rm prot}(t,t_0,t_M),\,\psi_{\rm prem}(t,t_0,t_M)\in\R^{n+m}$ are given by
\begin{align*}
\psi_{\rm prot}(t,t_0,t_M) & = (1-\delta)\left(\psi_{\rm D}(t,t_M) - \psi_{\rm D}(t,t_0) \right),\\
\psi_{\rm prem}(t,t_0,t_M) & = \sum_{j=1}^M (t_j - t_{j-1})\psi_{\rm Z}(t,t_j) +  \psi_{\rm D_*}(t,t_M) - \psi_{\rm D_*}(t,t_0) \\
& \quad + t_{M-1} \psi_{\rm D}(t,t_M) - \sum_{j=1}^{M-1} (t_j - t_{j-1}) \psi_{\rm D}(t,t_j) - t_0 \psi_{\rm D}(t,t_0).
\end{align*}
\end{proposition}
As a consequence of Proposition~\ref{prop:cds}, the CDS spread is given by a readily available linear-rational expression,
\begin{equation}\label{eq:CDSspread}
{\rm CDS}(t,t_0,t_M) = \Ind{\tau>t}\frac{\psi_{\rm prot}(t,t_0,t_M)^\top\begin{pmatrix} Y_t \\ X_t \end{pmatrix}}{\psi_{\rm prem}(t,t_0,t_M)^\top \begin{pmatrix} Y_t \\ X_t \end{pmatrix}}.
\end{equation} 
This is a remarkably simple expression that allows us to see how the factors $(Y,X)$ affect the CDS spread through the vectors $\psi_{\rm prot}(t,t_0,t_M)$ and $\psi_{\rm prem}(t,t_0,t_M)$.
For comparison, in an affine default intensity model the two legs $V_{\rm prot}(t,t_0,t_M)$ and $V_{\rm prem}(t,t_0,t_M)$ are given as sums of exponential-affine terms that cannot be simplified further.
In the following, we denote $V_{\rm CDS}(t,t_0,t_M,k)$ the time-$t$ price of a CDS contract starting at time $t_0$ with maturity $t_M$ and spread $k$,
\begin{equation}\label{eq:Vcds}
V_{\rm CDS}(t,t_0,t_M,k) = \Ind{\tau>t} \big( \psi_{\rm prot}(t,t_0,t_M) - \psi_{\rm prem}(t,t_0,t_M)  \big)^\top \begin{pmatrix} Y_t \\ X_t \end{pmatrix}.
\end{equation}

A \emph{multi-name CDS}, or credit default index swap (CDIS), is an insurance on a reference portfolio of $N$ firms with equal weight which we assume to be $1/N$ so that the portfolio total notional amount is equal to one.
The protection buyer pays a regular premium that is proportional to the current notional amount of the CDIS.
Let $\delta\in[0,1]$ be the recovery rate determined at inception.
Upon default of a firm the protection seller pays $1-\delta$ to the protection buyer and the notional amount of the CDIS decreases by $1/N$.
These steps repeat until maturity or until all firms in the reference portfolio have defaulted, whichever comes first.

Denote by $S_t^i=a_i^\top Y_t$ the survival process of firm $i$ as defined in~\eqref{eq:S}. The CDIS spread simplifies to a double linear-rational expression,
\begin{align*}
{\rm CDIS}(t,t_0,t_M) &= \frac{\sum_{i=1}^N \Ind{\tau_i>t} \, (1/a_i^\top Y_t) \, \psi^i_{\rm prot}(t,t_0,t_M)^\top\begin{pmatrix} Y_t \\ X_t \end{pmatrix}}{\sum_{i=1}^N \Ind{\tau_i>t} \, (1/a_i^\top Y_t) \, \psi^i_{\rm prem}(t,t_0,t_M)^\top \begin{pmatrix} Y_t \\ X_t \end{pmatrix}}
\end{align*}
where $\psi^i_{\rm prot}(t,t_0,t_M)$ and $\psi^i_{\rm prem}(t,t_0,t_M)$ are defined as in Proposition~\ref{prop:cds} for each firm $i$. 

\begin{remark}[Unspanned Factors]
The characteristics of the martingales $M^Y$ and $M^X$ do not appear explicitly in the bond, CDS and CDIS pricing formulas. This leaves the freedom to specify exogenous factors that feed into $M^Y$ and $M^X$. Such factors would be unspanned by the term structures of defaultable bonds and CDS and give rise to unspanned stochastic volatility, as described in \cite{filipovic2017linear}.
They provide additional flexibility for fitting time series of bond prices and CDS spreads.
These unspanned stochastic volatility factors affect the distribution of the survival and factor processes and therefore can be recovered from the prices of credit derivatives such as those discussed hereinafter.
\end{remark}

\subsection{CDIS tranche} \label{sec:tranche}

A CDIS tranche is a partial insurance on the losses of a reference portfolio in the sense that only losses larger than the attachment point $K_a$ and lower than the detachment point $K_d$ are insured.
We assume the same tenor structure and reference portfolio as for the CDIS contract, the protection buyer pays a periodic premium that is proportional to the current notional amount of the tranche, 
\begin{equation}\label{eq:trancheDef}
T_t = (K_d-K_a-(N_t(1-\delta)/N -K_a)^+)^+
\end{equation}
where $N_t=\sum_{i=1}^N \Ind{\tau_i \le t}$ is the total number of firms which have defaulted in the reference portfolio at time $t$.
The values of the protection leg and the premium leg at time $t$ are respectively given by
\begin{equation} \label{eq:tranchelegs}
\begin{split}
V_{\rm prot }(t,t_M,K_a,K_d) & =\int_t^{t_M} \e^{-r u} \; \cExp{  dT_u}{\Gcal_t}, \text{ and}\\
V_{\rm prem}(t,t_M,K_a,K_d) &= \sum_{j=1}^M \e^{-rt_j} \int_{t_{j-1}}^{t_j} \big( K_d-K_a- \cExp{ T_u}{\Gcal_t} \big) du.
\end{split}
\end{equation} 
The value of the tranche is then simply given by the difference of the cash flow values,
\begin{equation}\label{eq:tran_price}
V_{\rm T}(t,t_M,K_a,K_d,k) = V_{\rm prot }(t,t_M,K_a,K_d) - k \, V_{\rm prem}(t,t_M,K_a,K_d)
\end{equation}
where $k$ is the tranche spread.
The following proposition shows that the ${(\Fcal_\infty\vee\Gcal_t)}$-conditional distribution of the number of defaults at time $u>t$, can be exactly retrieved in closed form by applying the discrete Fourier transform as described in~\cite{ackerer2016dependent}.
\begin{proposition} \label{prop:distNt}
The $(\Fcal_\infty\vee\Gcal_t)$-conditional distribution of the number of defaults $N_u$ is given by
\begin{equation} \label{eq:distNt}
\Q [ N_u=n \mid \Fcal_\infty \vee \Gcal_t]=
\frac{1}{N+1}\sum_{j=0}^N \; \zeta^{nj} \; \prod_{i=1}^N \Big( \zeta^j +(1-\zeta^j)\Ind{\tau_i>t}\frac{a_i^\top Y_u}{a_i^\top Y_t}\Big)
\end{equation}
for $u>t$, for any $n=0,\dots,N$, and where $\zeta=\exp(2\im\pi/(N+1))$ with the imaginary number $\im$.
\end{proposition}
From~\eqref{eq:trancheDef} it follows immediately that the conditional expectation of $T_u$ can be expressed as a function of the conditional distribution of $N_u$.
Assume for simplicity that $K_a=n_a (1-\delta)/N$ and $K_d=n_d (1-\delta)/N$ for some integers $0\le n_a<n_d\le N$. Then the conditional expectation of $T_u$ is given by
\begin{equation}\label{eq:TrDist}
\cExp{T_u}{\Fcal_\infty\vee\Gcal_t} =
 \sum_{j=1}^{N-n_a} \frac{(1-\delta) \min(j, \, n_d-n_a) }{N} \; \Q[N_u=n_a+j \mid \Fcal_\infty\vee\Gcal_t]
\end{equation}
for $u>t$.

The tranche price~\eqref{eq:tran_price} has therefore a closed form expression as long as the conditional probability ${\Q[N_u=j\mid\Gcal_t]}$ is available in closed form for all $t\le u \le t_M$ and $j=0,\dots,N$. 
An example is given in Section~\ref{sec:tranche_ex} for a polynomial model.

\subsection{CDS option and CDIS option} \label{sec:CDopt}

A CDS option with strike spread $k$ is a European call option on the CDS contract exercisable only if the firm has not defaulted before the option maturity date $t_0$. 
Its payoff at time $t_0$ is given by
$$
(V_{\rm CDS}(t_0,t_0,t_M))^+
 = \frac{\Ind{\tau>t_0}}{a^\top Y_{t_0}}\Big( \psi_{\rm cds}(t_0,t_0,t_M,k)^\top \begin{pmatrix} Y_{t_0} \\ X_{t_0} \end{pmatrix} \Big)^+
$$
with
\begin{equation}\label{eq:psicds}
\psi_{\rm cds}(t,t_0,t_M,k)  = \psi_{\rm prot}(t,t_0,t_M)-k \, \psi_{\rm prem}(t,t_0,t_M).
\end{equation}
Denote by $V_{\rm CDSO}(t,t_0,t_M,k)$ the price of the CDS option at time $t$,
\begin{align*}
V_{\rm CDSO}(t,t_0,t_M,k) 
&= \E \Big[ \e^{-r(t_0-t)} \frac{\Ind{\tau>t_0}}{a^\top Y_{t_0}} \Big( \psi_{\rm cds}(t_0,t_0,t_M,k)^\top \begin{pmatrix} Y_{t_0} \\ X_{t_0} \end{pmatrix} \Big)^+ \; \Big| \; \Gcal_{t}\Big] \\
&= \Ind{\tau>t}\frac{\e^{-r(t_0-t)}}{a^\top Y_{t}}\E \Big[  \Big( \psi_{\rm cds}(t_0,t_0,t_M,k)^\top \begin{pmatrix} Y_{t_0} \\ X_{t_0} \end{pmatrix} \Big)^+ \; \Big| \; \Fcal_{t} \Big]
\end{align*}
where the second equality follows directly from Lemma~\ref{lem:fundpricing}.

A CDIS option gives the right at time $t_0$ to enter a CDIS contract with strike spread $k$ and maturity $t_M$ on the firms in the reference portfolio which have not defaulted and, simultaneously, to receive the losses realized before the exercise date $t_0$.
Denote by $V_{\rm CDISO}(t,t_0,t_M,k)$ the price of the CDIS option at time $t \le t_0$,
$$
V_{\rm CDISO}(t,t_0,t_M,k) =  \frac{\e^{-r(t_0-t)}}{N} \E \Big[ \Big( \sum_{i=1}^N V^i_{\rm CDS}(t_0,t_0,t_M,k) + (1-\delta)\Ind{\tau_i\le t_0} \Big)^+ \; \Big| \; \Gcal_t \Big],
$$
where $V^i_{\rm CDS}(t_0,t_0,t_M,k)$ is defined as in~\eqref{eq:Vcds} for firm $i$.

\begin{proposition}\label{prop:CDISOprice}
The price of a CDIS option is given by
$$
V_{\rm CDISO}(t,t_0,t_M,k)= \sum_{\alpha\in\{0,1\}^N} \frac{\e^{-r(t_0-t)}}{N}  \cExp{ (V_{*}(\alpha,t_0,t_M,k))^+ q(\alpha,t, t_0)}{\Fcal_t}
$$
with the conditional payoffs
$$
V_{*}(\alpha,t_0,t_M,k) = \sum_{i=1}^N \; \frac{\alpha_i}{a_i^\top Y_{t_0}} \psi^i_{\rm cds}(t_0,t_0,t_M,k)^\top \begin{pmatrix} Y_{t_0} \\ X_{t_0} \end{pmatrix} + (1-\delta)(1-\alpha_i) 
$$
and the conditional probabilities
$$
q(\alpha,t,t_0) = \prod_{i=1}^N  \;  \frac{(a_i^\top Y_{t_0})^{\alpha_i}(a_i^\top (Y_t-Y_{t_0}))^{1-\alpha_i}}{a_i^\top Y_{t}} \Ind{\tau_i>t} + (\Ind{\tau_i\le t})^{1-\alpha_i}
$$
where $\alpha=(\alpha_1,\dots,\alpha_N)$ and with the convention $0^0=0$.
\end{proposition}

The time-$t$ price of a CDS option, or of a CDIS option, is therefore given by the expected value of a non-smooth continuous function in $(Y_{t_0},X_{t_0})$ where $t<t_0$.
A methodology to price such contracts is presented in Section~\ref{sc:lhcm:optapp}.

\subsection{Credit valuation adjustment} \label{sec:CVA}

The unilateral credit valuation adjustment (UCVA) of a position in a bilateral contract is the present value of losses resulting from its cancellation when the counterparty defaults.

\begin{proposition}
The time-$t$ price of the UCVA with maturity $t_M$ and time-$u$ net positive exposure $f(u,Y_u,X_u)$, for some continuous function $f(u,y,x)$, is
\begin{align*}
{\rm UCVA}(t,t_M) &= \cExp{ \e^{-r(\tau-t)} \Ind{t<\tau\le t_M}f(\tau,Y_{\tau},X_{\tau} )}{\Gcal_t} \\
&= \frac{\Ind{\tau>t}}{a^\top Y_{t}}\int_t^{t_M} \e^{-r(u-t)}\cExp{f(u,Y_u,X_u)\,a^\top(c\, Y_u + \gamma \, X_u)}{\Fcal_t}du .
\end{align*}
where $\tau$ is the counterparty default time.
\end{proposition}
Computing the UCVA therefore boils down to a numerical integration of European style option prices. As is the case for CDS and CDIS options, these option prices can be uniformly approximated as described in Section~\ref{sc:lhcm:optapp}.
We refer to~\cite{brigo2014arbitrage} for a thorough analysis of bilateral counterparty risk valuation in a doubly stochastic default framework.

\section{The linear hypercube model}\label{sec:LHCM}

The linear hypercube (LHC) model is a single-name model, that is $n=1$ so that $S=Y$.
The survival process is absolutely continuous, as in Remark~\ref{remhara}, and the factor process $X$ is diffusive and takes values in a hypercube whose edges' length is given by $Y_t$, for all $t\ge0$.
More formally the state space of $(Y,X)$ is given by
\begin{equation*}
E = \big\{ (y,x)\in\R^{1+m} \, : \, y\in(0,1]\text{ and } x \in [0,y]^m \big\}.
\end{equation*}
The dynamics of $(Y,X)$ is
\begin{equation}\label{SXdyn}
\begin{aligned}
dY_t & = -\gamma^\top X_t \,dt \\
dX_t & = ( b Y_t + \beta X_t)\,dt + \Sigma(Y_t,X_t)\,dW_t
\end{aligned}
\end{equation}
for some $\gamma\in\R^m_+$ and some $m$-dimensional Brownian motion $W$, and where the volatility matrix $\Sigma(y,x)$ is given by
\begin{equation} \label{Xdisp}
\Sigma(y,x)={\rm diag}\big( \sigma_1\sqrt{x_1(y-x_1)} ,   \dots,\, \sigma_m\sqrt{x_m( y-x_m)} \big)
\end{equation}
with volatility parameters $\sigma_1, \dots,\, \sigma_m\ge 0$.

Let $(Y,X)$ be an $E$-valued solution of \eqref{SXdyn}. It is readily verified that $Y$ is non-increasing and that the parameter $\gamma$ controls the speed at which it decreases,
\[ 0 \le \gamma^\top X_t \le \gamma^\top \bm 1 \, Y_t,  \]
which implies
\begin{equation}\label{APbound}
0\le \lambda_t\le \gamma^\top \bm 1 \quad \text{and} \quad Y_t \ge Y_0 \, \e^{-\gamma^\top \bm 1 \, t} > 0, \quad\text{for any } t\ge0.
\end{equation}
Note that the default intensity upper bound $\gamma^\top\bm1$ depends on $\gamma$, which is estimated from data.
Therefore, a crucial step in the model validation procedure is to verify that the range of possible default intensities is sufficiently wide.

The following theorem gives conditions on the parameters such that the LHC model~\eqref{SXdyn} is well defined.

\begin{theorem}\label{thmexiuniLCRM}
Assume that, for all $i=1,\dots,\,m$,
\begin{align}
 b_i  - \sum_{j\neq i}   \beta_{ij}^-  &\ge  0 ,\label{eq:condatzero}\\
\gamma_i + \beta_{ii} +  b_i  + \sum_{j\neq i}  ( \gamma_j + \beta_{ij}) ^+ &\le 0.  \label{eq:condatLS}
\end{align}
Then for any initial law of $(Y_0,X_0)$ with support in $E$ there exists a unique in law $E$-valued solution $(Y,X)$ of~\eqref{SXdyn}. It satisfies the boundary non-attainment, for any $i=1,\dots,m$,
\begin{enumerate}
\item $X_{it}>0$ for all $t\ge 0$ if $X_{i0}>0$ and
\begin{equation}
b_i  - \sum_{j\neq i}   \beta_{ij}^-  \ge  \frac{\sigma_i^2}{2} ,\label{eq:condatzeroS}
\end{equation}

\item $X_{it}<Y_t$ for all $t\ge 0$ if $X_{i0}<Y_0$ and
\begin{equation} \label{eq:condatLSS}
\gamma_i + \beta_{ii} +  b_i  + \sum_{j\neq i}  ( \gamma_j + \beta_{ij}) ^+ \le - \frac{\sigma_i^2}{2}.
\end{equation}
\end{enumerate}
\end{theorem}

The state space $E$ is a regular $(m+1)$-dimensional hyperpyramid. Figure~\ref{fig:1Fss} shows $E$ when $m=1$ and illustrates the drift inward pointing conditions~\eqref{eq:condatzero}--\eqref{eq:condatLS} at the boundaries of $E$.

In Section~\ref{sec:MPR} we describe all possible market price of risk specifications under which the drift function of $(Y,X)$ remains linear.

\begin{remark}
The volatility of $X_{i}$ is maximal at the center of its support when ${X_{i}=Y/2}$ and decreases to zero at its boundaries for $X_{i}\to$ 0 and $X_{i}\to Y$.
As a consequence, some factors may alternate visits to the lower part and upper part of their supports, and therefore may mimic regime-shifting behavior.
\end{remark}

\begin{remark}
Define the normalized process $Z = X/Y$, then the dynamics of $(Z, \lambda)$ is given by
\begin{align}
dZ_t &= \big(b +  (\beta + \diag(\gamma^\top Z_t) )Z_t \big)dt + \Sigma(1, \, Z_t)dW_t, \\
d\lambda_t &= \gamma^\top dZ_t.
\end{align}
We derive closed form expressions for the stationary points of the drift of $(Z,\lambda)$ in Sections~\ref{sec:lhc1n} and~\ref{sec:fit}, and in Example~\ref{ex:negcor}.
\end{remark}

\subsection{One-factor LHC model} \label{sec:lhc1n}

The default intensity of the one-factor LHC model, $m=1$, has autonomous dynamics of the form
$$
d\lambda_t = ( \lambda_t^2 + \beta \lambda_t + b \gamma )dt + \sigma\sqrt{\lambda_t(\gamma-\lambda_t)}dW_t.
$$
The diffusion function of $\lambda$ is the same as the diffusion function of a Jacobi process taking values in the compact interval $[0,\gamma]$.
However, the drift of $\lambda$ includes a quadratic term that is neither present in Jacobi nor in affine processes.\footnote{The Jacobi process has been used in \cite{delbaen2002interest} to model the short rate in which case the risk-free bond prices are given by weighted series of Jacobi polynomials in the short rate value.}
Conditions \eqref{eq:condatzero}--\eqref{eq:condatLS} in Theorem~\ref{thmexiuniLCRM} rewrite
\[
b \ge 0 \quad\text{and}\quad    (\gamma + b + \beta) \le 0.\]
That is, the drift of $\lambda$ is nonnegative at $\lambda=0$ and nonpositive at $\lambda=\gamma$.
We can factorize the drift as follows,
\[  \lambda_t^2 + \beta \lambda_t + b\gamma = (\lambda_t-\ell_1)(\lambda_t-\ell_2), \]
for some roots $0\le \ell_1\le \gamma\le \ell_2$. 
Hence $\lambda$ drifts towards $\ell_1$, as long as not ${\lambda_t=\ell_2=\gamma}$.
The corresponding original parameters are given by $\beta=-(\ell_1+\ell_2)$ and ${b\gamma=\ell_1\ell_2}$, so that the drift of the factor $X$  reads
\begin{equation}\label{mrcase}
  \beta Y_t+B X_t  = (\ell_1+\ell_2)\bigg(\frac{\ell_1\ell_2}{\gamma(\ell_1+\ell_2)} Y_t -X_t\bigg).
\end{equation}

As a sanity check we verify that the constant default intensity case, $\lambda_t = \gamma$ for all $t\ge0$, is nested as a special case.
This is equivalent to have $X= Y$, which can be obtained by specifying the dynamics $dX_t = -\gamma X_t\,dt $ for the factor process and the initial condition $X_0=1$.
This corresponds to the stationary points $\ell_1=0$ and $\ell_2=\gamma$.

The dynamics of the standard one-factor affine model on $\R_+$ is
\begin{equation*}
d\lambda_t = \ell_2(\ell_1-\lambda_t)dt +  \sigma\sqrt{\lambda_t}dW_t,
\end{equation*}
where $\ell_2$ is the mean-reversion speed and $\ell_1$ the mean-reversion level of $\lambda$.
Figure~\ref{fig:1Fdynamic} shows the drift and diffusion functions of the default intensity for the one-factor LHC and affine models. The drift function is affine in the affine model whereas it is quadratic in the LHC model.
However, for reasonable parameters values, the drift functions look similar when the default intensity is smaller than the mean-reversion level $\lambda<\ell_1$.
On the other hand when $\lambda>\ell_1$, the force of drifting towards $\ell_1$ is smaller and concave in the LHC model. The diffusion function is strictly increasing and concave for the affine model whereas it has a concave semi-ellipse shape in the LHC model.
The diffusion functions have the same shape on $[0,\gamma/2]$ but typically do not scale equivalently in the parameter $\sigma$.
Note that the parameter $\gamma$ can always be set sufficiently large so that the likelihood of $\lambda$ going above $\gamma/2$ is arbitrarily small.

\subsection{Option price approximation} \label{sc:lhcm:optapp}

We saw in Sections~\ref{sec:CDopt} and \ref{sec:CVA} that the pricing of a CDS option, a CDIS option, or a UCVA boils down to computing a $\Fcal_t$-conditional expectation of the form
\[ \Phi(f;t,t_M)=\cExp{f(Y_{t_M},X_{t_M})}{\Fcal_t} \]
for some continuous function $f(y,x)$ on $E$. We now show how to approximate $\Phi(f;t,t_M)$ in closed form form by means of a polynomial approximation of $f(y,x)$.
The methodology presented hereinafter applies to any linear credit risk model which has a compact state space $E$ and for which the $\Fcal_t$-conditional moments of $(Y_{t_M},X_{t_M})$ are computable.

To this end, we first recall how the $\Fcal_t$-conditional moments of $(Y_{t_M},X_{t_M})$ for $t\le t_M$ can be obtained in closed form as described in \cite{filipovic2016polynomial}.
Denote by $\text{Pol}_n(E)$ the set of polynomials $p(y,x)$ on $E$ of degree $n$ or less. 
It is readily seen that the generator of $(Y,X)$,
\begin{align}\label{eqgen}
 \Gcal f(y,x)&= \big( -\gamma^\top x \quad (\beta y+Bx)^\top \big) \nabla f(y,x)\\
& \quad +\frac{1}{2} \sum_{i=1}^m \frac{\partial^2 f(y,x)}{\partial x_i^2} \sigma_i^2 x_i(y-x_i),
\end{align}
is polynomial in the sense that
\[ \Gcal \text{Pol}_n(E) \subset \text{Pol}_n(E)  \quad \text{for any $n\in\N$.} \]
Let $N_n=\binom{n+1+m}{n}$ denote the dimension of $\text{Pol}_n(E)$ and fix a polynomial basis $\{h_1,\dots,h_{N_n}\}$  of $\text{Pol}_n(E)$. We define the function of $(y,x)$
\[ H_n(y,x)=(h_1(y,x),\dots,h_{N_n}(y,x))^\top \]
with values in $\R^{N_n}$.
There exists a unique matrix representation $G_n$ of $\Gcal\mid_{\text{Pol}_n(E) }$ with respect to this polynomial basis such that for any $p\in\text{Pol}_n(E) $ we can write
\[ \Gcal p(y,x) = H_n(y,x)^\top \, G_n \, \vec{p}\]
where $\vec p$ is the coordinate representation of $p$. 
This implies the moment formula
\begin{equation}\label{PPP}
\cExp{p(Y_{t_M},X_{t_M})}{\Fcal_t} = H_n(Y_t,X_t)^\top \e^{G_n (t_M-t)} \vec{p}
\end{equation}
for any $t\le t_M$ we have, see~\cite[Theorem~(3.1)]{filipovic2016polynomial}.

\begin{remark}
  The choice for the basis $H_n(y,x)$ of $\text{Pol}_n(E)$ is arbitrary and one may simply consider the monomial basis,
\[ H_n(y,x)=\{1,\,y,\,x_1,\dots,\,x_m,\,y^2,\,yx_1,\,x_1^2,\dots , \, x_m^n\}\]
in which $G_n$ is block-diagonal.
There are efficient algorithms to compute the matrix exponential $\e^{G_n (t_M-t)}$, see for example \cite{Higham2008}.
Note that only the action of the matrix exponential is required, that is $\e^{G_n (t_M-t)}\vec{p}$ for some $p\in\text{Pol}_n(E)$, for which specific algorithms exist as well, see for examples \cite{al2011computing} and \cite{sidje1998expokit} and references within.
\end{remark}

Now let $\epsilon>0$. From the Stone-Weierstrass approximation theorem \cite[Theorem 5.8]{rudin1974functional} there exists a polynomial $p\in{\rm Pol}_n(E)$ for some $n$ such that
\begin{equation}\label{WT}
\sup_{(y,x)\in E}\left\lvert f(y,x) - p(y,x) \right\lvert \le \epsilon.
\end{equation}
Combining \eqref{PPP} and \eqref{WT} we obtain the desired approximation of $\Phi(f;t,T)$.

\begin{theorem}
Let $p\in{\rm Pol}_n(E)$ be as in \eqref{WT}. Then $\Phi(f;t,t_M)$ is uniformly approximated by
\begin{equation}\label{eqEB}
\sup_{t\le t_M}\left\| \Phi(f;t,t_M)- H_n(Y_t,X_t)^\top \e^{G_n (t_M-t)} \vec{p}\right\|_{L^\infty}\le \epsilon.
\end{equation}
\end{theorem}

The approximating polynomial $p$ in \eqref{WT} needs to be found case by case.  We illustrate this for the CDS option in Section~\ref{sec:CDSopt} and for the CDIS option on an homogenous portfolio in Section~\ref{sec:CDISopt}. 

\begin{remark}
Approximating the payoff function $f(y,x)$ on a strict subset of the state space $E$ is sufficient to approximate an option price.
Indeed, for any times $t \le u \le s $ the process $(Y_u,X_u)_{t\le u\le s}$ takes values in
$$
\big\{(y,x)\in E \,:\, Y_t \ge y \ge \e^{-\gamma^\top\bm 1(s-t)}Y_t \big\} \subset E.
$$
A polynomial approximation on a compact subset of $E$ can be expected to be more precise and, as a result, to produce a more accurate price approximation. 
See Section~\ref{sec:CDSopt} for an implementation example.
\end{remark}

\section{Case studies}\label{sec:app}

We show that the LHC model can reproduce complex term structure dynamics, that option prices can be accurately approximated, and that the prices of derivatives on homogeneous portfolios can similarly be computed.
First, we fit a parsimonious LHC model specification to CDS data and discuss the estimated parameters and factors. Then, we accurately approximate the price of CDS options at different moneyness.
Finally, for a homogeneous portfolio, we derive closed form expressions for the payoff function of a CDIS option and for the tranche prices.

\subsection{CDS calibration} \label{sec:fit}

We calibrate the LHC model to a high yield firm, Bombardier Inc., and also to an investment grade firm, Walt Disney Co., in order to show that the model flexibly adjusts to different spread levels and dynamics.
We also present a fast filtering and calibration methodology which is specific to LHC models.

\paragraph{Data description.}
The empirical analysis is based on composite CDS spread data from Markit which are essentially averaged quotes provided by major market makers.
The sample starts on January 1\textsuperscript{th} 2005 and ends on January 1\textsuperscript{th} 2015.
The data set contains 552 weekly observations summing up to 3620 observed CDS spreads for each firm.
At each date we include the available spreads with the modified restructuring clause on contracts with maturities of 1, 2, 3, 4, 5, 7, and 10 years.

Time series of the 1-year, 5-year, and 10-year CDS spreads are displayed in Figure~\ref{fig:cdsts}, as well as the relative changes on the 5-year versus 1-year CDS spread.
The two term-structures of CDS spreads exhibit important fluctuations of their level, slope, and curvature.
The time series can be split into three time periods.
The first period, before the subprime crisis, exhibits low spreads in contango and low volatility.
The second period, during the subprime crisis, exhibits high volatility with skyrocketing spreads temporarily in backwardation.
The crisis had a significantly larger impact on the high yield firm for which the spreads have more than quadrupled.
The third period is characterized by a steep contango and a lot of volatility. Figure~\ref{fig:cdsts} also shows that CDS spread changes are strongly correlated across maturities.
Summary statistics are reported in Table~\ref{tab:cdsstats}.

\paragraph{Model specification.}
The risk neutral dynamics of each survival process is given by the LHC model of Section~\ref{sec:LHCM} with two and three factors.
We set $\gamma =  \gamma_1 \bm e_1$, for some $\gamma_1\ge  0$, and consider a cascading structure of the form
\begin{equation} \label{eq:LHCCi}
dX_{it} = \kappa_{i}(\theta_i X_{(i+1)t}-X_{it})\,dt + \sigma_i\sqrt{X_{it}(Y_t-X_{it})}\,dW_{it}
\end{equation}
for $i=1,\dots,\,m-1$ and
\begin{equation} \label{eq:LHCCm}
dX_{mt} = \kappa_{m}(\theta_m Y_t-X_{mt})\,dt + \sigma_m\sqrt{X_{mt}(Y_t-X_{mt})}\,dW_{mt}
\end{equation}
for some parameters $\kappa,\theta,\sigma\in\R^m_+$ satisfying
\begin{equation}\label{eq:LHCCconst}
\theta_i \le
1 - \frac{\gamma_1}{\kappa_i} 
\end{equation}
for $i=1,\dots,\,m$.
We have that $\beta_{ii}=-\kappa_i$, $\beta_{i,i+i}=\kappa_i\theta_i$, and $\beta_{ij}=0$ otherwise,  $b_m=\kappa_m \theta_m$ and $b_i=0$ otherwise.
It directly follows that
\[
0 \le b_i - \sum_{j\neq i} \beta_{ij}^- = \Ind{i=m} \kappa_m \theta_m = \Ind{i=m} \beta_{mm} 
\]
and for $i=1,\dots,\,m$
\begin{align*}
0  & \ge \gamma_i + \beta_{ii} +  b_i  + \sum_{j\neq i}  ( \gamma_j + \beta_{ij}) ^+  = \gamma_1 -\kappa_i + \kappa_i\theta_i\\
&  \quad =\gamma_1+\beta_{ii}+\Ind{i\neq m}\beta_{i,i+1} + \Ind{i=m}b_m.
\end{align*}
This shows that the parameter conditions~\eqref{eq:condatzero}-\eqref{eq:condatLS} are satisfied.
Note that~\eqref{eq:condatzero}-\eqref{eq:condatLS} boil down to standard linear parameter constraints when expressed in terms of $\beta$ and $b$.
They are therefore compatible with efficient optimization algorithms.

This specification allows default intensity values to persistently be close to zero over extended periods of time.
It also allows to work with a multidimensional model parsimoniously as the number of free parameters is equal to $3m+1$ whereas it is equal to $3m + m^2$ for the generic LHC model.
The default intensity is then proportional to the first factor and given by $\lambda=\gamma_1X_{1}/Y$.

We denote the two-factor and three-factor linear hypercube cascade models by ${\rm LHCC }(2)$ and ${\rm LHCC }(3)$, respectively.
In addition, we estimate a three-factor model, denoted ${\rm LHCC }(3)^*$, where parameter $\gamma_1$ is an exogenous fixed parameter.
This parameter value is fixed so as to be about twice as large as the estimated $\gamma_1$ from the ${\rm LHCC }(3)$ model.
We estimate the constrained model in order to determine whether the choice of the default intensity upper bound is critical for the empirical results.

We set the risk-free rate equal to the average 5-year risk-free yield over the sample, $r=2.52\%$.
We make the usual assumption that the recovery rate is equal to $\delta=40\%$.
We also use Lemma~\ref{lem:simplerFormulas} to compute efficiently the CDS spreads, which is justified by the following result.

\begin{lemma}\label{lhccA*inv}
Assume that $r>0$, then the matrix $A^*=A-r\Id$ with $A$ as in~\eqref{eq:Amat} is invertible for the cascade LHCC model defined in~\eqref{eq:LHCCi}--\eqref{eq:LHCCm} and with $\gamma=\gamma_1\e_1$.
\end{lemma}

\begin{remark}
The drift of the normalized process $Z=X/Y$ admits the stationary points $\bar{\mu}_t$ given by the system of equations
\begin{equation}\label{eq:lhcc_mu}
\bar{\mu}_{it} = (-1)^{m-i+1}\prod_{j=i}^m \frac{\kappa_j\theta_j}{\bar{\mu}_{1t}\gamma_1 - \kappa_j}, \quad i=1,\dots,m
\end{equation}
as shown in Appendix~\ref{sec:LCRMproofs}. 
In fact, $\bar{\mu}_{1t}$ implies the values of $\bar{\mu}_{it}$ for $i=2,\dots,m$.
The stationary point of the drift of $\lambda$ is given by $\gamma_1\bar{\mu}_{1t}$.
\end{remark}

\paragraph{Filtering and calibration.}

We present an efficient methodology to filter the factors from the CDS spreads.
We recall that the CDS spread ${\rm CDS}(t,t_0,t_M)$ is the strike spread that renders the initial values of the CDS contract equal to zero. We therefore obtain the affine equation
\begin{equation}
\label{eq:ZtEq}
\psi_{\rm cds}(t,t_0,t_M,{\rm CDS}(t,t_0,t_M))^\top \begin{pmatrix} 1 \\ Z_t \end{pmatrix} = 0
\end{equation}
conditional on $\tau>t$, and with the normalized process $Z = X / Y \in [0,1]^m$.
Therefore, in theory we could extract the value $Z_t$ from the observation of at least $m$ spreads with different maturities.
The factor value $(S_t,X_t)$ at time $t$ can in turn be inferred, for example, by applying the Euler scheme to compute the survival process value, for example, and then rescaling the pseudo factor  $Z_t$,
\begin{equation}\label{eq:filter:S}
Y_{t_i} = Y_{t_{i-1}} - \gamma^\top  X_{t_{i-1}} \Delta t \quad \text{and} \quad X_{t_i} = Y_{t_i} Z_{t_i} 
\end{equation}
for all the observation dates $t_i$, and with $Y_{t_0}=1$.
In practice, there might not be a value $Z_t$ such that \eqref{eq:ZtEq} is satisfied for all observed market spreads. Therefore, we consider all the observable spreads and minimize the following weighted mean squared error
\begin{equation}\label{eq:filter:Z}
\begin{aligned}
& \underset{z}{\min} 
& &  \frac{1}{2} \; \sum_{k=1}^{n_{i}} \; \Bigg(\, \frac{ \psi_{\rm cds}(t_i,t_i,t_M^k,{\rm CDS}(t_i,t_i,t_M^k))^\top \begin{pmatrix} 1 \\ z \end{pmatrix}}{\psi_{\rm prem}(t_{i},t_i,t_M^k)^\top \begin{pmatrix} 1 \\ Z_{t_{i-1}} \end{pmatrix}}
\,\Bigg)^2\\
& \text{s.t.}
& & 0\le z_i  \leq 1, \; i = 1, \ldots, m
\end{aligned}
\end{equation}
where $t_M^1,\dots,t_M^{n_i}$ are the maturities of the $n_{i}$ observed spreads at date $t_i$, and $t_{i-1}$ is the previous observation date.
Dividing the CDS price error by an approximation of the CDS premium leg value gives an accurate approximation of the CDS spread error when $Z_{t_i}\approx Z_{t_{i-1}}$.
The above minimization problem is a linearly constrained quadratic optimization problem which can be solved virtually instantaneously numerically.

For any parameter set we can extract the observable factor process at each date by recursively solving~\eqref{eq:filter:Z} and applying~\eqref{eq:filter:S}.
With the parameters and the factor process values we can in turn compute the difference between the model and market CDS spreads.
Therefore, we numerically search the parameter set that minimizes the aggregated CDS spread root-mean-squared-error (RMSE) by using the gradient-free Nelder-Mead algorithm together with a penalty term to enforce the parameter constraints and starting from several randomized initial parameter sets.

Note that we do not calibrate the volatility parameters $\sigma_i$ for $i=1,\dots,m$ since CDS spreads do not depend on the martingale components with linear credit risk models and since the factor process is observable directly from the CDS spreads.
Furthermore, we only fit the risk-neutral drift parameters $\kappa$ and $\theta$ implied by the CDS spreads.
The total number of parameters for LHCC(2), LHCC(3), and LHCC(3)$^*$ model is therefore equal to 5, 7, and 6 respectively. 
Equipped with a fast filter and a low dimensional parameter space, the calibration procedure is swift.

\begin{remark}
Alternatively one could estimate the parameters by performing a quasi-maximum likelihood estimation or a more advanced generalized method of moments estimation.
This can be implemented in a straightforward manner with the LHC model if the market price of risk specification preserves the polynomial property of the factors as the real--world conditional moments of $(Y,X)$ are then given in closed form, see Appendix~\ref{sec:MPR}.
The availability of conditional moments also enables direct usage of the Unscented Kalman Filter to recover the factor values at each date.
However this approach comes at the cost of more parameters and possibly more stringent conditions on them, as well as unnecessary computational costs if we are only interested in market prices.
\end{remark}

\paragraph{Parameters, fitted spreads, and factors.}

The fitted parameters are reported in Table~\ref{tab:paramsLCRM}.
An important observation is that the parameter constraint in~\eqref{eq:LHCCconst} is binding for each dimension in all the fitted models.
The calibrated parameter values are similar across the different specifications which is comforting, and the calibrated default intensity upper bounds appear large enough to cover the high spread values observed during the subprime crisis.

The fitted factors extracted from the calibration are used as input to compute the fitted spreads.
With the fitted spreads we compute the fitting errors for each date and maturity.
Not surprisingly the more flexible specification ${\rm LHCC }(3)$ performs the best.
Estimating the default intensity upper bound $\gamma_1$ instead of setting an arbitrarily large value improves the calibration.
Table~\ref{tab:rmse} reports summary statistics of the errors by maturity.
The ${\rm LHCC }(3)$ model has the smallest RMSE for each maturity.
In particular, its overall RMSE is half the one of the two-factor model.
The ${\rm LHCC }(3)^*$ model faces difficulties in reproducing long-term spreads as, for example, its RMSE is twice as large as the one of the unconstrained ${\rm LHCC }(3)$ for the 10-year maturity spread for both firms.
Figure~\ref{fig:fitLCRM} displays the fitted spreads and the RMSE time series.
Again, the ${\rm LHCC }(3)$ appears to have the smallest level of errors over time.
The two other models do not perform as good during the low spreads period before the financial crisis, and during the recent volatile period.
Overall, the fitted models appear to reproduce relatively well the observed CDS spread values.

Figure~\ref{fig:factors} shows the estimated factors. They are remarkably similar across the different specifications.
The default intensity explodes and the survival process decreases rapidly during the financial crisis.
The $m$-th factor controls the long term default intensity level.
The second factors controls the medium term behavior of the term-structure of credit risk in the  ${\rm LHCC }(3)$ and  ${\rm LHCC }(3)^*$  models.
The  ${\rm LHCC }(2)$  model requires an almost equal to zero default intensity to capture the steep contango of the term structure at the end of the sample period, even lower than before the financial crisis.
This seems counterfactual and illustrates the limitations of the  ${\rm LHCC }(2)$  model in capturing changing dynamics.
The $m$-th factor visits the second half of its support $[0,Y_t]$ and appears to stabilize in this region for the three models.

\subsection{CDS option pricing} \label{sec:CDSopt}

We describe an accurate and efficient methodology to price CDS options that builds on the payoff approximation approach presented in Section~\ref{sc:lhcm:optapp}, and illustrate it with numerical examples.
The model used for the numerical illustration is the one-factor LHC model from Section~\ref{sec:lhc1n} with stylized but realistic parameters $\gamma=0.25$, $\ell_1=0.05$, $\ell_2=1$, $\sigma=0.75$, $X_0=0.2$, and $r=0$.

From Section~\ref{sec:CDopt}, we know that the time-$t$ CDS option price with strike spread $k$ is of the form
\[
V_{\rm CDSO}(t,t_0,t_M,k) = \Ind{\tau>t} \cExp{f(Z(t_0,t_M,k))}{\Fcal_t}
\]
with the payoff function $f(z) = \e^{-r(t_0-t)} z^+ / Y_t$ and where the random variable $Z(t_0,t_M,k)$ is defined by
\begin{equation} \label{eq:Zdef}
Z(t_0,t_M,k)=\psi_{\rm cds}(t_0,t_0,t_M,k)^\top \begin{pmatrix}
  Y_{t_0} \\ X_{t_0}
\end{pmatrix}
\end{equation}
with $\psi_{\rm cds}(t_0,t_0,t_M,k)$ as in~\eqref{eq:psicds}.
Furthermore, the random variable $Z(t_0,t_M,k)$ takes values in the interval $[b_{min},b_{max}]$ with the LHC model which is given by 
\begin{align*}
b_{min} &= \sum_{i=1}^{m+1} \min(0,\psi_{\rm cds}(t_0,t_0,t_M,k)_i), \text{ and}\\
b_{max} &= \sum_{i=1}^{m+1} \max(0,\psi_{\rm cds}(t_0,t_0,t_M,k)_i).
\end{align*}
We now show how to approximate the payoff function $f$ with a polynomial by truncating its Fourier-Legendre series, and then how the conditional moments of $Z(t_0,t_M,k)$ can be computed recursively from the conditional moments of $(Y_{t_0},X_{t_0})$.

Let $\Lcal e_n(x)$ denote the generalized Legendre polynomials taking values on the closed interval $[b_{min},b_{max}]$ and given by
\[
\mathcal{L}e_n(x) = \sqrt{\frac{1+2n}{2\sigma^2}} \, Le_{n}\Big(\frac{x-\mu}{\sigma}\Big)
\]
where $\mu=(b_{max}+b_{min})/2$, $\sigma=(b_{max}-b_{min})/2$, and the standard Legendre polynomials $Le_n(x)$ on $[-1,1]$ are defined recursively by 
\[
Le_{n+1}(x) = \frac{2n+1}{n+1} x \,Le_n(x) - \frac{n}{n+1} Le_{n-1}(x)
\]
with $Le_0=1$ and $Le_1(x)=x$.
The generalized Legendre polynomials form a complete orthonormal system on $[b_{min},b_{max}]$ in the sense that the mean squared error of the Fourier-Legendre series approximation $f^{(n)}(x)$ of any piecewise continuous function $f(x)$, defined by
\begin{equation}\label{eq:payoffapp}
f^{(n)}(x) = \sum_{k=0}^n \, f_n \, \Lcal e_n(x), \quad \text{where} \; f_n = \int_{b_{min}}^{b_{max}} \, f(x) \, \Lcal e_n(x) \, dx,
\end{equation}
converges to zero, 
\[
\lim_{n\rightarrow \infty} \; \int_{b_{min}}^{b_{max}} \big(f(x) - f^{(n)}(x) \big)^2 dx = 0.
\]
The coefficients for the CDS option payoff are given in closed form,
\[
f_n = \Ind{\tau>t} \frac{\e^{-r(t_0-t)}}{Y_t} \int_0^{b_{max}} z \, \Lcal e_n(z) \, dz,
\]
since the integrands are polynomial functions.
Note that a similar approach is followed in~\cite{ackerer2016jacobi} on the unbounded interval $\R$ with a Gaussian weight function.

The $\Fcal_t$-conditional moments of $Z(t_0,t_M,k)$ can be computed recursively from the conditional moments of $(Y_{t_0},X_{t_0})$.
Let $\pi:\Ecal\mapsto\{1,\dots,N_{n}\}$ be an enumeration of the set of exponents with total order less or equal to $n$, that is
$$
\Ecal = \big\{ \bm{\alpha} \in\N^{1+m} \,:\, \sum_{i=1}^{1+m} {\alpha}_i \le n \big\}.
$$
Define the polynomials 
$$
h_{\pi(\bm{\alpha})}(s,x)=s^{\alpha_1} \, \prod_{i=1}^m x_i^{\mathbf{\alpha}_{1+i}},
$$
which form a basis of ${\rm Pol}_n(E)$.
Denote by $\bm{1}$ the $(1+m)$-dimensional vector of ones and by $\e_i$ the $(1+m)$-dimensional vector whose $i$-th coordinate is equal to one and zero otherwise.
\begin{lemma}\label{lem:coefsrec}
For all $n\ge 2$ we have
\[
\cExp{Z(t_0,t_M,k)^n}{\Fcal_t} = \sum_{\bm{\alpha}^\top \bm{1}=n} c_{\pi(\bm{\alpha})} \; \cExp{h_{\pi(\bm \alpha)}(Y_{t_0},X_{t_0})}{\Fcal_t}
\]
where the coefficients $c_{\pi(\bm{\alpha})}$ are recursively given by
\[
c_{\pi(\bm{\alpha})} = \sum_{i=1}^{1+m} \Ind{\alpha_i-1\ge0} \, c_{\pi(\bm{\alpha}-\e_i)} \,\psi_{\rm cds}(t_0,t_0,t_M,k)_i.
\]
\end{lemma}

We now report the main numerical findings.
We take $t_0=1$, $t_M=t_0+5$, and three reference strike spreads $k\in\{250,300,350\}$ basis points that represent in, at, and out of the money CDS options.
The first row in Figure~\ref{fig:CDSO_payoff} shows the payoff approximation $f^{(n)}(z)$ in~\eqref{eq:payoffapp} for the polynomial orders $n\in\{1,5,30\}$ and the strike spreads $k\in\{250,300,350\}$.
A more accurate approximation of the hockey stick payoff function is naturally obtained by increasing the order $n$, especially around the kink.
The width of the support $[b_{min},b_{max}]$ increases with the strike spread $k$, hence the uniform error bound should be expected to be larger for out of the money options.
This is confirmed by the second row of Figure~\ref{fig:CDSO_payoff} that shows the error bound~\eqref{eqEB} as a function of the approximation order $n$ for the Fourier-Legendre approach described above.
It also displays the error bound when the CDS option payoff function is interpolated by means of Chebyshev  polynomials, see Appendix~\ref{sec:chebint} for more details.
The error bound is approximated by taking the maximum distance between the payoff function and the polynomial approximation on a regular grid of $10^4$ points over $[b_{min},b_{max}]$.
We remark that the error bound of the Chebyshev approach is oscillating around the error bound of the Fourier-Legendre approach.
This seems to be caused by variation of the polynomial approximation accuracy around the payoff kink as the Chebyshev nodes change.
Note that the error bound is typically non tight in practice, as illustrated in the following pricing application in which the pricing error is far lower than the error bound at least for $n\le20$.

Figure~\ref{fig:CDSO_pricecpu} shows the price approximation as a function of the polynomial order, up to $n=30$.
The price approximations stabilize rapidly with the Fourier-Legendre approach so that a price approximation using the first $n=10$ moments appear to be accurate up to a basis point.
On the other hand, the price approximations exhibit large oscillations with the Chebyshev approach.
Figure~\ref{fig:CDSO_pricecpu} also shows that it takes a fraction of a second on a standard desktop to compute the price approximation. Note that almost all of the CPU time is spent on the computation of the moments of $Z(t_0,t_M,k)$.

We recall that the volatility parameter $\sigma$ of the LHC model does not affect the CDS spreads, and can therefore be used to improve the joint calibration of CDS and CDS options. We illustrate this in the left panel of Figure~\ref{fig:CDSoptSens} where the CDS option price is displayed as a function of the volatility parameter for different strike spreads. As expected, the option price is an increasing function of the volatility parameter. The right panel of Figure~\ref{fig:CDSoptSens} also shows that $X_0$ has an almost linear impact on the CDS option price.

Note that the dimension of the polynomial basis $\binom{1+m+n}{n}$ becomes a programming and computational challenge when both the expansion order $n$ and the number of factors $1+m$ are large.
For example, for $n=20$ and $1+m=2$ the basis has dimension 231 whereas it has dimension 10'626 when $1+m=4$.
In practice, we successfully implemented examples with $1+m=4$ and $n=50$ on a standard desktop computer, in which case the basis dimension is 316'251.

\subsection{CDIS option pricing} \label{sec:CDISopt}

We discuss the approximation of the payoff function by means of Chebyshev polynomials for a CDIS option on a homogeneous portfolio.
Let ${N_t=\sum_{i=0}^N \Ind{\tau_i \le t}}$ denotes the number of firms which have defaulted by time $t$.
Consider a CDIS option on a homogeneous portfolio so that $S_t^i=a^\top Y_t$ for all $i=1,\dots,N$.
From Proposition~\ref{prop:CDISOprice} it follows that the time-$t$ price of the CDIS option is given by
$$
V_{\rm CDISO}(t,t_0,t_M,k) = \frac{\e^{-r(t_0-t)}}{N} \sum_{j=0}^{N-N_t}  \cExp{ (V_*(j,t_0,t_m))^+ \, q(j,t,t_0)}{\Fcal_t}
$$
with the conditional payoffs 
$$
V_*(j,t_0,t_m) = \frac{j}{a^\top Y_{t_0}} \, \psi_{\rm cds}(t_0,t_0,t_M,k)^\top \begin{pmatrix}    Y_{t_0} \\ X_{t_0} \end{pmatrix} + (1-\delta)(N-j)
$$
and the conditional probabilities
\begin{equation}\label{eq:N_surv_homo}
q(j,t,t_0) = \binom{N-N_t}{j} \, \frac{(a^\top Y_{t_0})^j(a^\top Y_t-a^\top Y_{t_0})^{N-N_t-j}}{(a^\top Y_t)^{N-N_t}}
\end{equation}
with the notable difference that now the summation contains at most $N+1$ terms because the defaults are symmetric and thus interchangeable.
Define the random variables
\[
Y(t_0) = a^\top Y_{t_0} \quad \text{and} \quad
X(t_0,t_M,k)= \psi_{\rm cds}(t_0,t_0,t_M,k)^\top \begin{pmatrix}    Y_{t_0} \\ X_{t_0} \end{pmatrix}.
\]
The CDIS option price then rewrites
\[
V_{\rm CDISO}(t,t_0,t_M,k) = \cExp{f(Y(t_0),X(t_0,t_M,k))}{\Fcal_t \vee N_t}
\]
where the bivariate payoff function $f(y,x)$ is given by
\begin{align*}
f(y,x) & = \frac{\e^{-r(t_0-t)}}{N \, (a^\top Y_t)^{N-N_t}} \bigg( (1-\delta) N (a^\top Y_t-y)^{N-N_t} \\
& \quad + \sum_{j=1}^{N-N_t} \binom{N-N_t}{j}  \left(j \, x +y(1-\delta)(N-j)\right)^+ y^{j-1}(a^\top Y_t-y)^{N-N_t-j} \bigg).
\end{align*}
The $\Fcal_t$-conditional moments of $(Y(t_0),\,X(t_0,t_M,k))$ can be computed recursively in a similar way as in Lemma~\ref{lem:coefsrec}.
The payoff function $f(y,x)$ can be approximated using Chebyshev polynomials and nodes, see~Appendix~\ref{sec:chebint}, or using its two-dimensional Fourier-Legendre series representation.

\subsection{CDIS tranche pricing}\label{sec:tranche_ex}

As in Section~\ref{sec:CDISopt}, we consider a homogeneous portfolio so that $S^i=a^\top Y$ for all $i=1,\dots,N$.
In this case, a simpler expression for~\eqref{eq:distNt} can be derived,
\begin{equation}\label{eq:NtHomo}
\begin{aligned}
\Q[ N_u=j \mid \Fcal_\infty\vee\Gcal_t]  &=  \Q[ N-N_u=N -j \mid \Fcal_\infty\vee\Gcal_t] \\
& = q(N-j,t,u)
\end{aligned}
\end{equation}
for $ u>t$ and $j=N_t,\dots,N$, and where $q(N-j,t,u)$ is defined as in~\eqref{eq:N_surv_homo}. 
We fix the attachment point to ${K_a = n_a(1-\delta)/N}$ and  the detachment point to ${K_d =  n_d(1-\delta)/N}$, for some integers $0\le n_a < n_d\le  N$.
Assuming for simplicity that $N_t \le n_a$, then from~\eqref{eq:TrDist} and~\eqref{eq:NtHomo} we obtain that
$$
\cExp{T_u}{\Fcal_\infty\vee\Gcal_t} = \sum_{{j=n_a+1}}^{N} \frac{(1-\delta) \min(j -n_a, \, n_d-n_a) }{N} \; q(N-j, t, u)
$$
and, by differentiating with respect to $u$,
\begin{align*}
\frac{d\cExp{T_u}{\Fcal_\infty\vee \Gcal_t}}{du} &= \sum_{j=n_a+1}^{N} \frac{(1-\delta) \min(j-n_a,\,n_d-n_a) }{N} \\
& \quad \times \binom{N-N_t}{N-j}  \frac{ (a^\top Y_{u})^{N-j-1}(a^\top Y_{t}-a^\top Y_{u})^{j-N_t-1}}{(a^\top Y_{t})^{N-N_t}}  \\ 
& \quad \times \big((N-j)a^\top Y_{t} - (N-N_t)a^\top Y_{u}\big) \, a^\top(c \, Y_{u} + \gamma \, X_u)
\end{align*}
for any $u>t$.
The protection and premium legs in~\eqref{eq:tranchelegs} can thus, in principle, be computed in closed form using the moments formula~\eqref{PPP}.

\section{Extensions} \label{sec:extensions}

We present several model extensions offering additional features.
We first construct multi-name models, then include stochastic interest rates possibly correlated with credit spreads, and conclude by discussing jumps and stochastic clocks to generate simultaneous defaults.

\subsection{Multi-name models} \label{sec:multiname}

We build upon the LHC model to construct multi-name models with correlated default intensities and which can easily accommodate the inclusion of new factors and firms.
This approach can be applied to other linear credit risk models, as long as they belong to the class of polynomial models.
We consider $n$ independent LHC processes
\begin{equation}\label{eq:YXiid}
(Y^{1},X^{1}), \dots, \, (Y^{n},X^{n})
\end{equation}
where each $(Y^{j},X^{j})$ is defined as in~\eqref{SXdyn}--\eqref{Xdisp}.
We defined the stacked processes ${Y=(Y^{1},\dots,Y^{n})}$ with $Y_0=\bm 1$ and ${X=(X^1,\dots,X^{n})}$ with $X_0\in[0,1]^m$ where $m=\sum_{j=1}^n m_j$.
We denote $E$ the state space of $(Y,X)$.

Let $h=(h^1,\dots,h^n)$ be the $\R^n_+$-valued process whose $j$-th component is given by
\begin{equation}\label{eq:ht}
h_t^j = \frac{{\gamma^j}^\top X_t^j}{Y_t^j}, \quad t\ge0
\end{equation}
where the vector $\gamma^j\in\R^{m_j}$ is the drift parameter of $Y^j$, see~\eqref{SXdyn}. 

\paragraph{Linear Construction.} 
The survival process of the firm $i=1,\dots,N$ can be defined as in~\eqref{eq:S}, $S^{i} = a_i^\top \, Y$,
for some vector $a_i\in\R^n_+$ satisfying $a^\top {\bm 1} =1$.
The corresponding default intensity $\lambda^{i}$ of firm $i$ is for all $t\ge0$ given by a weighted sum of $h$, that is $\lambda_t^{i} =  {w^i_t}^\top h_t$ with stochastic weights $w^i_{jt} = a_{ij}Y^j_t / S^i_t>0$ satisfying $\sum_{j=1}^d w_{jt}^i=1$.

\paragraph{Polynomial Construction.} 
Fix a degree $d$ and define the survival process of each firm $i=1,\dots,N$ by
$S_t^i = p_i(Y_t)$ for all $t\ge0$,
for some polynomial $p_i(y)\in{\rm Pol}_d([0,1]^n)$ which is componentwise non-increasing and positive on $[0,1]^n$, and such that $p_i({\bm 1})=1$.
Let $H_d(y,x)$ be a polynomial basis of ${\rm Pol}_d(E)$ stacked in a row vector and of the form 
$H_d(y,x)=  (H_d(y), \, H^*_d(y,x))$
where $H_d(y)$ is itself a polynomial basis of ${\rm Pol}_d([0,1]^n)$.
The survival process of firm $i$ then rewrites $S^i = a_i^\top \Ycal$ with the finite variation process $\Ycal=H_d(Y)$, the factor process $\Xcal=H_d^*(Y,X)$ and where the vector $a_i$ is given by the equation $p_i(y)=H_d(y) \, a_i$.
It follows from the polynomial property that the process $(\Ycal,\Xcal)$ has a linear drift as in~\eqref{eq:dY}--\eqref{eq:dX}, see~\cite[Theorem~4.3]{filipovic2017polynomial}.
The specific values for the drift of $(\Ycal,\Xcal)$ depend on the choice of the polynomial basis $H_d(y,x)$.

\begin{example} \label{exa:Spol}
Take $p(y)=y^{\alpha}=\prod_{i=1}^n y_i^{\alpha_{i}}$ for some $\alpha\in\N^n$, then the implied default intensity is a weighted sum $\lambda_t = \alpha^\top h_t$ with $h_t$ as defined in~\eqref{eq:ht}.
The weights are constant as opposed to the stochastic weights in the linear construction.
\end{example}

\begin{remark}
The dimension of $H_d(y,x)$ is $\binom{d+n+m}{d}$ and may be large depending on the values of $m+n$ and $d$. However, given that the pairs $(Y_t^i,X_t^i)$ in~\eqref{eq:YXiid} are independent, the conditional expectation of a monomial in $(Y_{u},X_{u})$ rewrites
\[
\E \Big[ \prod_{i=1}^n(Y^i_{u})^{\alpha_{i}}(X^i_{u})^{\beta_i} \; \Big|\; \Fcal_t\Big] = \prod_{i=1}^n \; \E \Big[ (Y^i_{u})^{\alpha_{i}}(X^i_{u})^{\beta_i} \; \Big|\; \Fcal_t\Big], \quad u>t,
\]
for some $\alpha_i\in\N$ and $\beta_i\in\N^{m_j}$ for all $i=1,\dots,n$.
Hence, to compute bonds and CDSs prices we only need to consider $n$ independent polynomial bases of total dimension equal to $\sum_{i=1}^n \binom{d+1+m_i}{d}$.
\end{remark}

\subsection{Stochastic interest rates} 

We include stochastic interest rates possibly correlated with credit spreads.
We denote the discount process $D_t = \exp(-\int_0^t r_s ds)$, for $t\ge0$, where $r_s$ is the short rate value at time $s$.
We specify that $D=a_r^\top Y$ for some vector $a_r\in\R^n$.
This is similar to the specification of the survival process of a firm, but we do not require that $D$ is non-increasing. 
That is, we allow for negative interest rates. 
We follow Section~\ref{sec:multiname} and let $H_2(y,x)$ be a polynomial basis of ${\rm Pol}_2(E)$ which defines a new linear credit risk model $(\Ycal,\Xcal)=(H_2(Y),H_2^*(Y,X))$ whose linear drift is given by a matrix $\Acal$ as in~\eqref{eq:Amat}.

\begin{proposition} \label{prop:DtSt}
The pricing formulas~\eqref{eq:Zbond}, \eqref{eq:CFdef}, and \eqref{eq:CFdefs} also apply with $(\Ycal_t,\Xcal_t)$ in place of $(Y_t,X_t)$, with $r=0$, by using the vector
\begin{equation*}
\psi_{\rm Z}(t,t_M)^\top = \begin{pmatrix}
a_{\rm Z}^\top & 0
\end{pmatrix}  \e^{ \Acal(t_M-t)}
\end{equation*}
where the vector $a_{\rm Z}$ is given by $H_2(y)^\top a_{\rm Z} = (a_r^\top y)(a^\top y)$, and the vectors
\begin{align*}
\psi_{\rm D}(t,t_M)^\top &=  a_{\rm D} ^\top 
\int_t^{t_M}  \e^{\Acal(s-t)}ds ,\\
\psi_{\rm D_*}(t,{t_M})^\top &=  a_{\rm D} ^\top 
 \int_t^{t_M}  s \,\e^{\Acal(s-t)}ds,
\end{align*}
where the vector $a_{\rm D}$ is given by 
$
H_2(y,x)\, a_{\rm D} = 
(a_r^\top y ) ( -a^\top (cy  \; \gamma x))$.
\end{proposition}

In practice it can be sufficient to consider a basis strictly smaller than $H_2(y,x)$, as the following example suggests.

\begin{example} \label{exa:DtSt}
Consider two independent LHC processes $(Y^j,X^j)$ with $m_j=1$ for $j\in\{1,2\}$, and consider the following linear credit risk model with stochastic interest rate given by,
\begin{equation*}\label{eq:psiZrate}
D_t = Y^1_t \quad \text{and} \quad S_t = \nu \, Y_t^1 + (1-\nu) \, Y_t^2, \quad \text{for all $t \ge0,$}
\end{equation*}
for some parameter $\nu\in(0,1)$.
The calculation of bond and CDS prices only requires the subbases
\[
H_0(y,x)  = \begin{pmatrix}
y_1^2 & y_1\,y_2
\end{pmatrix}, \quad
H_1(y,x)  = \begin{pmatrix}
y_1 x_1 & y_1 x_2 & x_1 y_2 & x_1^2 & x_1 x_2
\end{pmatrix},
\]
whose total dimension is $\dim((H_0(y,x), H_1(y,x))) = 7 < \dim({\rm Pol}_2(E))=15$. 
The drift term of the process $(H_0(Y,X),\,H_1(Y,X))$ is
\[
\Acal = \begin{pmatrix}
0 & 0 & -2 \gamma_1 & 0 & 0 & 0 & 0 \\
0 & 0 & 0 & -\gamma_2 & -\gamma_1 & 0 & 0 \\
b_1 & 0 & \beta_1 & 0 & 0 & -\gamma_1 & 0\\
0 & b_2 & 0 & \beta_2 & 0 & 0 & -\gamma_1 \\
0 & b_1 & 0 & 0 & \beta_1 & 0 & 0 \\
\sigma_1^2 & 0 & 2b_1 - \sigma_1^2 & 0 & 0 & 2 \beta_1 & 0 \\
0 & 0 & 0 & b_1 & b_2 & 0 & \beta_1 + \beta_2\\
\end{pmatrix}
\]
where the subscripts indicate the LHC model identity.
The pricing vectors in this basis are
\[
a_{\rm Z} = \begin{pmatrix}
\nu & 1-\nu 
\end{pmatrix} 
\quad \text{and} \quad
a_{\rm D} = \begin{pmatrix}
0 & 0 & -\nu \, \gamma_1 & -(1-\nu) \, \gamma_2 & 0 & 0 & 0 
\end{pmatrix}.
\]

\end{example}

\subsection{Jumps and simultaneous defaults}

There are two ways to include jumps in the survival process dynamics that may result in the simultaneous default of several firms.
The first is to let the martingale part of $Y$ be driven by a jump process so that multiple survival processes may jump at the same time.
The second is to let time run with a stochastic clock leaping forward hence producing synchronous jumps in the factors and the survival processes.

The survival process remains defined as in~\eqref{eq:S} but the factors are extensions of the LHC process in what follows.
For simplicity, we discuss a unique pair $(Y, X)$ as in~\eqref{SXdyn} whose parameters $\gamma,\beta,B$ satisfy~\eqref{eq:condatzero}--\eqref{eq:condatLS}.
Let $Z$ be a nondecreasing L\'evy process with L\'evy measure $\nu^Z(d\zeta)$ and drift $b^Z\ge 0$ that is independent from the Brownian motion $W$ and the uniform random variables $U^1,\dots,U^N$. 

\paragraph{Jump-Diffusion Model.}
Assume that $\Delta Z_t \le 1$ for all $t\ge0$. We define the dynamics of the LHC model with jumps as follows
\begin{align*}
d\begin{pmatrix}
Y_t \\ X_t
\end{pmatrix} & = \begin{pmatrix}
-c  & -\gamma^\top - \delta^\top\E[Z_1] \\ b & \beta-\diag(\nu)\,\E[Z_1]
\end{pmatrix} 
\begin{pmatrix}
Y_{t-} \\ X_{t-}
\end{pmatrix} dt 
+ \begin{pmatrix}
0 \\ \Sigma(Y_{t-},X_{t-})
\end{pmatrix}dW_t \\
& \quad -  \begin{pmatrix}
c\,Y_{t-} + \delta^\top X_{t-} \\ \diag(\nu) X_{t-}
\end{pmatrix} dN_t
\end{align*}
with the martingale $N$ given by $N_t = Z_t - \E[Z_1] t$ for $t\ge0$, for some $c>0$, $\delta\in\R^m_+$, and $\nu\in\R^m_+$ such that
\begin{align}
c + \delta^\top{\bm 1}<1, \quad c+ \delta^\top{\bm 1} \le \nu_i \le 1, \quad \quad i=1,\dots,m \label{eq:condjump}\\
\text{and $\nu_i < 1$ if \eqref{eq:condatzeroS} applies,}\quad i=1,\dots,m \label{eq:condjumpat0}
\end{align}
Conditions~\eqref{eq:condjump}--\eqref{eq:condjumpat0} ensure that the process always jumps inside its state space.
Note that the same process $Z$ can affect the dynamics of multiple LHC processes $(Y^i,X^i)$.

\paragraph{Stochastic Clock.} 
We consider the time-changed process $(\bar Y_t, \bar X_t)_{t\ge 0}=(Y_{Z_t},X_{Z_t})_{t\ge 0}$ that will directly feed into~\eqref{eq:S} in place of $(Y_t,X_t)$ and whose factor dynamics is given by
\[
\begin{pmatrix}
d\bar Y_t \\ d\bar X_t
\end{pmatrix} = \bar{A} \begin{pmatrix}
\bar Y_t \\ \bar X_t
\end{pmatrix} dt + \begin{pmatrix}
dM^{\bar{Y}}_t \\ dM^{\bar{X}}_t
\end{pmatrix}
\]
where the $(m+n)\times(m+n)$-matrix $\bar{A}$ is now given by
\begin{equation} \label{eq:AmatTC}
\bar{A} = b^Z \, A + \int_0^\infty ( \e^{A\zeta} - \Id )\nu^Z (d\zeta)
\end{equation}
with the matrix $A$ as in Equation~\eqref{eq:Amat}, see \cite[Chapter 6]{sato1999levy} and \cite[Theorem~6.1]{filipovic2017polynomial}.
The time-changed LHC model remains a linear credit risk model.
The background filtration $\F$ is now the natural filtration of the process $(Y_{Z},X_{Z})$.
Denote $\Psi(\cdot)$ the Laplace exponent of $Z$ defined by $\mathbb{E}[\exp(-u Z_t)]=\exp(-t\Psi(u))$.
The following Proposition shows that the matrix $\bar A$ may be computed in closed form\footnote{We thank an anonymous referee for suggesting this result.}.
\begin{proposition}\label{prop:lhc_zt}
Assume that $A=UDU^{-1}$ where $U$ is a unitary matrix and $D$ is a diagonal matrix with nonpositive entries, then $\bar{A} = -U \Psi(-D) U^{-1}$.
\end{proposition}
In some cases, the expression for $\bar A$ simplifies and does not require factoring the matrix $A$ as shown in the following example.
\begin{example}
Let $Z$ be a Gamma process such that $\nu^Z(d\zeta)=\gamma_Z \zeta^{-1} \e^{-\lambda_Z \zeta}d\zeta$ for some constants $\lambda_Z,\gamma_Z>0$ and $b^Z=0$. If the eigenvalues of the matrix $A$ have nonpositive real parts, the drift of the time changed process $(Y_{Z},X_{Z})$ is then equal to
\begin{equation}\label{eq:AmatGammaTC}
\bar{A} = -\gamma_Z \, \log \big(\Id - A \lambda_Z^{-1} \big)
\end{equation}
as shown in Appendix~\ref{sec:LCRMproofs}.
\end{example}

Survival processes built from independent LHC models can be time changed with the same stochastic clock $Z$ in order to generate simultaneous defaults and thus default correlation.
Note that the idea of using time change to generate simultaneous jumps in the cumulative hazard or the survival processes is not new, see for example~\cite{mendoza2016multivariate} for an earlier contribution where a multi-name unified credit-equity model with simultaneous defaults is developed.

\begin{remark}
One could use the Additive subordinators presented in~\cite{li2016additive} in order to increase the model's flexibility. 
These subordinators are time-dependent and may therefore help to better fit term structures at the cost of introducing additional parameters.
In this case, the drift of the factor process $(\bar Y, \bar X)$ remains linear but the matrix $\bar A$ in~\eqref{eq:AmatTC} may then be time-dependent and may not have a closed form representation which would in turn lead to higher computational costs.
\end{remark}

\section{Conclusion}\label{sec:ccl}

The class of linear credit risk models is rich and offers new modeling possibilities.
The survival process and its drift are linear in the factor process whose drift is also linear.
Consequently, the prices of defaultable bonds, credit default swaps (CDSs), and credit default index swaps (CDISs) become linear-rational expressions in the factors.
We introduce and study the single-name linear hypercube (LHC) model which consists of a diffusive factor process with a quadratic diffusion function and taking values in a compact state space.
These features are employed to develop an efficient European option pricing methodology.
By building upon the LHC model, we construct parsimonious and versatile multi-name models.
The setup can accommodate stochastic interest rates correlated with credit spreads by constructing the discount process similarly as a survival process.
Jumps in the factor dynamics as well as stochastic clocks can be used to generate simultaneous defaults.
An empirical analysis shows that the LHC model can reproduce complex CDS term structure dynamics.
We numerically verify that CDS option prices at different moneyness can be accurately approximated for the LHC model.
We also show that CDIS option prices and tranche prices on a homogeneous portfolio can be approximated with the same approach.
Future research directions include the development of efficient algorithms to price multi-name credit derivatives, and the joint empirical study of single-name and multi-name credit contracts.


\begin{appendix}
\appendix\normalsize

\section{Proofs} \label{sec:LCRMproofs}

This Appendix contains the proofs of all theorems and propositions in the main text.

\subsection*{Proof of \eqref{lineq}}
This follows as in \cite[Lemma 3]{filipovic2017linear}.

\subsection*{Proof of Example~\ref{ex:negcor}}

The autonomous process $X$ admits a solution taking values in $[-e^{-\epsilon t}, e^{-\epsilon t}]$ at time $t$ with $\epsilon>0$ and $X_0\in[-1,1]$ if and only if $\kappa>\epsilon$, see~\cite[Theorem~5.1]{filipovic2016polynomial}.
The two coordinates of $Y$ are lower bounded by $X$.
Indeed for $i=1,2$ we have
\[
dY_{it} = -\frac\epsilon2(Y_{it} \pm X_t) dt \ge -\frac\epsilon2(Y_{it} + \e^{-\epsilon t})dt
\]
The solution of $dZ_t =  -(\epsilon/2)(Z_t + \e^{-\epsilon t})dt$ with $Z_0=1$ is given by $Z_t=\e^{-\epsilon t}$, $t\ge0$, which proves that $Y_{it}\ge Z_t \ge |X_t|$ for $i=1,2$.
Finally, by applying Ito's lemma we obtain
\[
d\langle\lambda^1,\lambda^2\rangle_t = - \frac{\epsilon^2}{4} \frac{\sigma^2(\e^{-\epsilon t}-X_t)(\e^{-\epsilon t}+X_t)}{Y_{1t}Y_{2t}}, \quad t\ge0,
\]
which is negative with positive probability.
The dynamics of $\lambda^i$ is given by,
\begin{align*}
d\lambda^i_t & = ({\epsilon^2}/{4}) \big(\pm(1 - 2\kappa/\epsilon) ({X_t}/{Y_{it}}) + ({X_t}/{Y_{it}})^2\big) dt \pm dM_{it} \\
& = \big( ({\epsilon}/{2}) (1 - 2\kappa/\epsilon) (\lambda^i_t - \epsilon/2) + (\lambda^i_t - \epsilon/2)^2 \big) dt \pm dM_{it}
\end{align*}
where $dM_{it} = \epsilon\,\sigma/(2Y_{it})\sqrt{(\e^{-\epsilon t}-X_t)(\e^{-\epsilon t}+X_t)}dW_t$, and $\kappa>\epsilon$. 
The quadratic drift of $\lambda^i$ has two positive roots, $\kappa$ and $\epsilon/2$, is positive at zero, and is negative at $\epsilon$.
Since $\kappa>\epsilon$, this shows that $\lambda^i$ mean reverts towards $\epsilon/2$ for $i=1,2$.

\subsection*{Proof of Proposition~\ref{prop:bondMat}}

Proposition~\ref{prop:bondMat} is an immediate consequence of \eqref{lineq} and the following lemma.
\begin{lemma}\label{lem:fundpricing}
Let $Y$ be a nonnegative $\Fcal_\infty$-measurable random variable.
For any  time ${t\le t_M<\infty}$,
\begin{equation*}
\cExp{\Ind{\tau>t_M} Y}{\Gcal_t} = \Ind{\tau>t} \frac{1}{S_t}\cExp{S_{t_M} Y}{\Fcal_t}.
\end{equation*}
Note that $t_M<\infty$ is essential unless we assume that $S_\infty = 0$.
\end{lemma}

Lemma~\ref{lem:fundpricing} follows from \cite[Corollary 5.1.1]{Bielecki2002credit}. For the convenience of the reader we provide here a sketch of its proof. As in \cite[Lemma 5.1.2]{Bielecki2002credit} one can show that, for any nonnegative random variable $Z$, we have
\begin{equation*}
\cExp{\Ind{\tau>t}Z}{\Hcal_t \vee \Fcal_t} = \Ind{\tau>t}\frac{1}{S_t}\cExp{\Ind{\tau>t} Z}{\Fcal_t}.
\end{equation*}
Setting $Z=\Ind{\tau>t_M}Y$ we can now derive
\begin{align*}
\cExp{\Ind{\tau>t_M}Y}{\Gcal_t} &= \cExp{\Ind{\tau>t} Y \Ind{\tau>t_M} }{\Gcal_t}  =\Ind{\tau>t}\frac{1}{S_t}\cExp{\Ind{\tau>t_M}Y}{\Fcal_t} \\
&= \Ind{\tau>t}\frac{1}{S_t}\cExp{\cExp{\Ind{\tau>t_M}}{\Fcal_\infty}Y}{\Fcal_t} \\
&= \Ind{\tau>t}\frac{1}{S_t}\cExp{S_{t_M} Y}{\Fcal_t}.
\end{align*}

\subsection*{Proof of Proposition~\ref{prop:bondDef}}

The subsequent proofs build on the following lemma that follows from \cite[Proposition~5.1.1]{Bielecki2002credit}.
\begin{lemma}\label{lem:fundpricing2}
Let $Z$ be a bounded $\F$-predictable process. 
For any ${t\le {t_M}<\infty}$,
\begin{equation*}
\cExp{\Ind{t<\tau\le {t_M}}Z_{\tau}}{\Gcal_t} = \Ind{t<\tau} \frac{1}{S_t}\int_{(t,{t_M}]}\cExp{ -Z_u dS_u}{\Fcal_t}.
\end{equation*}
Note that ${t_M}<\infty$ is essential unless we assume that $S_\infty = 0$.
\end{lemma}

We can now proceed to the proof of Proposition~\ref{prop:bondDef}. The value of the contingent cash flow is given by the expression
\begin{align*}
C_{\rm D}(t,t_M) & = \cExp{\e^{-r(\tau-t)}\Ind{t\le\tau\le t_M}}{\Gcal_t}
\end{align*}
By applying Lemma~\ref{lem:fundpricing2} we get
\begin{align*}
C_{\rm D}(t,t_M) & = \frac{\Ind{\tau>t}}{S_t}\int_t^{t_M}\cExp{ -\e^{-r(s-t)}dS_s}{\Fcal_t} \\
&= \frac{\Ind{\tau>t}}{S_t}\int_t^{t_M} \e^{-r(s-t)}\cExp{-a^\top(cY_s +\gamma X_s)}{\Fcal_t}ds \\
&= \frac{\Ind{\tau>t}}{S_t}\int_t^{t_M} \e^{-r(s-t)} -a ^\top \begin{pmatrix} c & \gamma \end{pmatrix} \e^{A(s-t)} \begin{pmatrix} Y_t \\ X_t \end{pmatrix} ds
\end{align*}
where the second equality comes from the fact that $\int_0^t \e^{-ru}\,dM^S_u$ is a martingale. The third equality follows from~\eqref{lineq}. 

\subsection*{Proof of Corollary~\ref{cor:condBond}}
The value of this contingent bond is given by
\begin{align*}
C_{\rm D_*}(t,{t_M}) & = \cExp{\tau\,  \e^{-r(\tau-t)}\Ind{t<\tau\le {t_M}}}{\Gcal_t} = \frac{\Ind{\tau>t}}{S_t}\int_t^{t_M} \cExp{ -s \, \e^{-r(s-t)}dS_s}{\Fcal_t} 
\end{align*}
and the result follows as in the proof of Proposition~\ref{prop:bondDef}.

\subsection*{Proof of Lemma~\ref{lem:simplerFormulas}}

Observe that for any matrix $A$ and real $r$ we have $\e^{r}\e^{A}=\e^{\diag(r)+A}$, and that the matrix exponential integration can be computed closed form as follows
\begin{align*}
\int_0^u \e^{As} ds & =  \int_0^u (I + As + A^2 \frac{s^2}{2}+ \dots) ds  =  Iu + A\frac{u^2}{2} + A^2 \frac{u^3}{6}+ \dots \\
& =  A^{-1}(\e^{Au} -I ).
\end{align*}
By change of variable $u=s-t$ we obtain
\[ \int_t^{t_M} s \e^{A_*(s-t)}ds = \int_0^{{t_M}-t}u\e^{A_*u}du + t \int_0^{{t_M}-t}\e^{A_*u}du, \]
where the second term on the RHS is given in Lemma~\ref{prop:bondDef}.
The first term can be derived using integration by parts
\begin{align*}
\int_0^{{t_M}-t} u\e^{A_*u}du &= ({t_M}-t) A_*^{-1}\e^{A_*({t_M}-t)} -  A_*^{-1} A_*^{-1} (\e^{A_*({t_M}-t)} - I).
\end{align*}

\subsection*{Proof of Proposition~\ref{prop:cds}}

The calculation of the protection leg and the coupon part, $V^i_{\rm prot}(t,t_0,t_M)$  and $V^i_{\rm coup}(t,t_0,t_M)$ respectively, follows from Propositions~\ref{prop:bondMat}~and~\ref{prop:bondDef}.
The accrued interest $V^i_{\rm ai}(t,t_0,t_M)$ is given by the sum of contingent cash flows and of weighted zero-recovery coupon bonds, and thus its calculation follows from Propositions~\ref{prop:bondDef} and~\ref{cor:condBond}.
The series of contingent cash flow is in fact equal to a single contingent payment paying $\tau$ at default,
\begin{align*}
C_{\rm D_*}(t,t_M) &= \sum_{j=1}^M \; \cExp{ \tau \, \e^{-r(\tau-t)}  \Ind{t_{j-1}<\tau\leq t_j}}{\Gcal_t}  \\
&=  \cExp{ \tau \, \e^{-r(\tau-t)}  \Ind{t<\tau\le t_M}}{\Gcal_t}.
\end{align*} 
Using the identity $\Ind{t_{j-1} < \tau \le t_j} = \Ind{\tau > t_{j-1}} - \Ind{\tau > t_{j}}$ we obtain that
the second term of $V^i_{\rm ai}(t,t_0,t_M)$ is given by
\begin{align*}
&-\sum_{j=1}^M \cExp{\e^{-r(\tau-t)} t_{j-1}\Ind{t_{j-1}<\tau\leq t_j}}{\Gcal_t} 
=\sum_{j=1}^M t_{j-1} \left(C_{\rm D}(t,t_{j}) - C_{\rm D}(t,t_{j-1}) \right) \\
 &= t_{M-1} C_{\rm D}(t,t_M) - T_{0} C_{\rm D}(t,t_0)
  - \sum_{j=1}^{M-1} (t_j-t_{j-1}) C_{\rm D}(t,t_j).
\end{align*}

\subsection*{Proof of Proposition~\ref{prop:distNt}}

The conditional characteristic function of $N_u$ is given by
\begin{align*}
\phi(t,\xi) &= \E \Big[ \exp\big( \im \xi N_u\big)\;\Big|\; \Fcal_\infty \vee\Gcal_t \Big] = \E\Big[\exp\big( \im \xi \sum_{i=1}^N \Ind{\tau_i\le u}\big) \;\Big|\; \Fcal_\infty \vee\Gcal_t \Big] \\
&= \E \Big[ \prod_{i=1}^N \big(\Ind{\tau_i>u} +  \e^{\im\xi}(1-\Ind{\tau_i>u}) \big) \; \Big| \; \Fcal_\infty \vee\Gcal_t \Big] \\
&= \prod_{i=1}^N \; \Big(\frac{\Ind{\tau_i>t}}{S_t^i} (S^i_u+  \e^{\im\xi}(S_t^i-S_u^i )) + \Ind{\tau_i\le t} \e^{\im\xi} \Big) \\
& = \prod_{i=1}^N \;   \Big( \e^{\im\xi} + \Ind{\tau_i>t} (1- \e^{\im\xi}) \frac{S_u^i}{S_t^i} \Big)
\end{align*}
where the first equality in the third line follows from~\cite[Lemma~9.1.3]{Bielecki2002credit}, which gives the expression
\begin{equation}\label{eq:BRlem9.13}
\E\left[\Ind{\tau_1>t_0,\dots,\,\tau_N>t_0} \mid\Fcal_{t_0}\vee\Gcal_t\right] = \prod_{i=1}^N \Ind{\tau_i>t}\frac{S_{t_0}^i}{S_t^i}.
\end{equation}
The expression~\eqref{eq:distNt} then directly follows by applying the discrete Fourier transform, see~\cite[Section~3]{ackerer2016dependent} for more details.

\subsection*{Proof of Proposition~\ref{prop:CDISOprice}}

The payoff at time $t_0$ of the CDIS option can always be decomposed into $2^N$ terms by conditioning on all the possible default events
\begin{equation} \label{eq:qal}
q(\alpha) = \prod_{i=1}^N \Big( (\Ind{\tau_i>t_0})^{\alpha_i} +   (\Ind{\tau_i\le t_0})^{1-\alpha_i} \Big)
\end{equation}
for $\alpha\in\{0,1\}^N$, and with the convention $0^0=0$, so that the payoff function rewrites
\begin{align*}
&\Big( \sum_{i=1}^N \frac{\Ind{\tau_i>t_0}}{S^i_{t_0}} \, \psi^i_{\rm cds}(t_0,t_0,t_M,k)^\top \begin{pmatrix} Y_{t_0} \\ X_{t_0} \end{pmatrix}+ (1-\delta)\Ind{\tau_i\le t_0} \Big)^+ \\
&= \sum_{\alpha\in\{0,1\}^N}\Big( \sum_{i=1}^N \frac{\alpha_i}{S^i_{t_0}} \, \psi^i_{\rm cds}(t_0,t_0,t_M,k)^\top \begin{pmatrix} Y_{t_0} \\ X_{t_0} \end{pmatrix} + (1-\delta)(1-\alpha_i) \Big)^+ q(\alpha).
\end{align*}
We can apply \cite[Lemma~9.1.3]{Bielecki2002credit} to compute the probability~\eqref{eq:BRlem9.13} so that by writing~\eqref{eq:qal} as a linear combination of indicator functions we obtain
\begin{align*}
q(\alpha,t,t_0)&=\E\left[q(\alpha)\mid\Fcal_{t_0}\vee\Gcal_t\right] \\
&= \prod_{i=1}^N \bigg( \frac{(S^i_{t_0})^{\alpha_i}(S_t^i-S^i_{t_0})^{1-\alpha_i}}{S_t^i} \Ind{\tau_i>t} + (\Ind{\tau_i\le t})^{1-\alpha_i} \bigg)
\end{align*}
which completes the proof.

\subsection*{Proof of Theorem~\ref{thmexiuniLCRM}}
We define the bounded continuous map $(\Ycal,\Xcal):R^{1+m}\to R^{1+m}$ by
\[  \Ycal(y,x)=y^+\wedge 1,\quad
  \Xcal_i(y,x)=x_i^+\wedge y^+\wedge 1,\quad i=1,\dots,m,\]
such that $(\Ycal,\Xcal)(y,x)=(y,x)$ on $E$. In a similar vein, extend the dispersion matrix $\Sigma(y,x)$ to a bounded continuous mapping $\Sigma((\Ycal,\Xcal)(y,x))$ on $\R^{1+m}$. The stochastic differential equation~\eqref{SXdyn} then extends to $\R^{1+m}$ by
\begin{equation}\label{SXdynext}
\begin{aligned}
dY_t & = -\gamma^\top \Xcal(Y_t,X_t) \,dt \\
dX_t & = \left( b \Ycal(Y_t)+ \beta \Xcal(Y_t,X_t)\right) dt + \Sigma\left((\Ycal,\Xcal)(Y_t,X_t)\right)dW_t.
\end{aligned}
\end{equation}
Since drift and dispersion of \eqref{SXdynext} are bounded and continuous on $\R^{1+m}$, there exists a weak solution $(Y,X)$ of \eqref{SXdynext} for any initial law of $(Y_0,X_0)$ with support in $E$, see \cite[Theorem~V.4.22]{karatzas1991brownian}.

We now show that any weak solution $(Y,X)$ of \eqref{SXdynext} with $(Y_0,X_0)\in E$ stays in $E$,
\begin{equation}\label{claimEinv}
\text{$(Y_t,X_t)\in E$ for all $t\ge 0$.}
\end{equation}
To this end, for $i=1,\dots,m$, note that
\begin{equation}\label{Snull}
\text{$\Sigma_{ii}\left((\Ycal,\Xcal)(y,x)\right)= 0$ for all $(y,x)$ with $x_i\le 0$ or $x_i\ge y$.}
\end{equation}
Conditon \eqref{eq:condatzero} implies that
\begin{equation}\label{eq:condatzero1}
\text{$\left( b \Ycal(y)+ \beta \Xcal(y,x) \right)_i\ge 0$ for all $(y,x)$ with $x_i\le 0$.}
\end{equation}
For $\delta,\epsilon>0$ we define
\[ \tau_{\delta,\epsilon}=\inf\left\{ t\ge 0\mid \text{$X_{it}\le -\epsilon$ and $-\epsilon < X_{is}<0$ for all $s\in [t-\delta,t)$}\right\} .\]
Then on $\{\tau_{\delta,\epsilon}<\infty\}$ we have, in view of~\eqref{Snull} and \eqref{eq:condatzero1},
\[ 0> X_{i\tau_{\delta,\epsilon}} - X_{i\tau_{\delta,\epsilon}-\delta} = \int_{\tau_{\delta,\epsilon}-\delta}^{\tau_{\delta,\epsilon}}  \left(b \Ycal(Y_u)+ \beta \Xcal(Y_u,X_u) \right)_i du \ge 0,\]
which is absurd. Hence $\tau_{\delta,\epsilon}=\infty$ a.s.\ and therefore $X_{it}\ge 0$ for all $t\ge 0$. Similarly, conditon~\eqref{eq:condatLS} implies that
\begin{equation}\label{eq:condatLS1}
\text{$-\gamma^\top\Xcal(y,x)-\left(b \Ycal(y)+ \beta \Xcal(y,x)\right)_i\ge 0$ for all $(y,x)$ with $x_i\ge y$.}
\end{equation}
Using the same argument as above for $Y_t-X_{it}$ in lieu of $X_{it}$, and \eqref{eq:condatLS1} in lieu of \eqref{eq:condatzero1}, we see that $Y_t-X_{it}\ge 0$ for all $t\ge 0$. 
Note that $0\le \gamma^\top\Xcal(y,x)\le \gamma^\top\bm 1 y^+$ for all $(y,x)$, and thus $1\ge Y_t\ge \e^{-\gamma^\top\bm 1 t}>0$ for all $t\ge 0$. 
This proves \eqref{claimEinv} and thus the existence of an $E$-valued solution of \eqref{SXdyn}.

Uniqueness in law of the $E$-valued solution $(Y,X)$ of \eqref{SXdyn} follows from \cite[Theorem~4.2]{filipovic2016polynomial} and the fact that $E$ is relatively compact.

The boundary non-attainment conditions~\eqref{eq:condatzeroS}--\eqref{eq:condatLSS} follow from \cite[Theorem~5.7(i) and (ii)]{filipovic2016polynomial} for the polynomials $p(y,x)=x_i$ and $y-x_i$, for $i=1,\dots,m$.

\subsection*{Proof of Lemma~\ref{lhccA*inv}}

The matrix $A_*$ in the LHCC model is given by
\[ A_* = 
\begin{pmatrix}
-r & -\gamma_1 & 0 & 0 &  \\
0 & -(\kappa_{1} + r) & \kappa_{1}\theta_1 & 0 &\cdots \\
\vdots &  & & \ddots & \\ 
\theta_m & & & 0 & - (\kappa_{m}+r)
\end{pmatrix}
\]
and its determinant is therefore equal to
\begin{align*}
|A_*| &= -r \, \begin{vmatrix} -(\kappa_{1} + r) & \kappa_{1}\theta_1 & 0 &\cdots \\
\vdots & & \ddots & \\ 
0 & &  0& - (\kappa_{m}+r) \end{vmatrix}\\
& \quad + \, (-1)^m \,
\begin{vmatrix}
-\gamma_1 & 0 & 0 &  \\
-(\kappa_{1} + r) & \kappa_{1}\theta_1 & 0 &\cdots \\
\vdots  & & \ddots & \\ 
0 & & - (\kappa_{m}+r) & \kappa_{m} \theta_m
\end{vmatrix}.
\end{align*}
With $r>0$, the first element on the right hand side is nonzero with sign equal to $(-1)^{1+m}$ and the second element also has a sign equal to $(-1)^{1+m}$.
This is because the determinant of a triangular matrix is equal to the product of its diagonal elements.
As a result, the determinant of $A_*$ is nonzero which concludes the proof.

\subsection*{Proof of Equation~\eqref{eq:lhcc_mu}}

For $i=1,\dots,m$ we have that $d(1/Y_t) = \gamma_1 Z_{1t}/Y_t$ for all $t\ge0$. 
The dynamics of $Z$ is thus given by
$$
dZ_{it} = (\kappa_i \theta_i Z_{(i+1)t} - \kappa_i Z_{it} + \gamma_1 Z_{1t} Z_{it})dt + \sigma_i\sqrt{Z_{it}(1-Z_{it})}\,dW_{it},
$$
for $ i=1,\dots,m-1$, and by
$$
dZ_{mt} = (\kappa_m \theta_m - \kappa_m Z_{mt} + \gamma_1 Z_{1t} Z_{mt})dt + \sigma_m\sqrt{Z_{mt}(1-Z_{mt})}\,dW_{mt}
$$
Fixing $Z_{1t}=\bar{\mu}_{1t}$ and solving for the value of $Z_{mt}$ which cancels its drift we obtain
\[
\bar{\mu}_{mt} = \frac{-\kappa_m\theta_m}{\bar{\mu}_{1t} \gamma_1 - \kappa_m},
\] 
and solving recursively for $i=m-1,\dots,1$ gives~\eqref{eq:lhcc_mu}.

\subsection*{Proof of Lemma~\ref{lem:coefsrec}}

We $n$-th power of $Z(t_0,t_M,k)$ rewrites
\begin{align*}
Z(t_0,t_M,k)^n &= \Big(\psi_{\rm cds}(t_0,t_0,t_M,k)^\top \begin{pmatrix} Y_{t_0} \\ X_{t_0} \end{pmatrix} \Big)^n \\
& = \psi_{\rm cds}(t_0,t_0,t_M,k)^\top \begin{pmatrix} Y_{t_0} \\ X_{t_0} \end{pmatrix} \; \sum_{\mathbf{\alpha}^\top \bm 1=n-1} c_{\pi(\bm{\alpha})} \; h_{\pi(\bm \alpha)}(Y_{t_0},X_{t_0}) \\
&= \sum_{i=1}^{1+m} \sum_{\mathbf{\alpha}^\top \bm 1=n-1} c_{\pi(\bm{\alpha})} \psi_{\rm cds}(t_0,t_0,t_M,k)_i \; h_{\pi(\bm \alpha + \e_i)}(Y_{t_0},X_{t_0})
\end{align*}
which is a polynomial containing all and only polynomials of degree $n$, the lemma follows by rearranging the terms.

\subsection*{Proof of Proposition~\ref{prop:DtSt}}

The time-$t$ price of the zero-coupon zero-recovery bond is now given by
\begin{align*}
B_Z(t,{t_M}) &= \E \Big[ \frac{D_{t_M}}{D_t} \Ind{\tau>{t_M}} \; \Big| \; \Gcal_t \Big]
= \frac{\Ind{\tau>t}}{D_t S_t} \; \E \Big[ D_{t_M} S_{t_M} \; \Big| \;  \Fcal_t \Big]\\
&= \frac{\Ind{\tau>t}}{(a_r^\top Y_t)(a^\top Y_t)} \; \E \Big[ (a_r^\top Y_{t_M})(a^\top Y_{t_M}) \; \Big| \;  \Fcal_t \Big]  \\
&= \frac{\Ind{\tau>t}}{a_{Z}^\top \Ycal_t} \begin{pmatrix}
a_{\rm Z}^\top & 0
\end{pmatrix} \e^{\Acal({t_M}-t)}\begin{pmatrix}
\Ycal_t \\ \Xcal_t
\end{pmatrix}
\end{align*}
by applying Lemma~\ref{lem:fundpricing}.
Similarly for contingent cash flows by Lemma~\ref{lem:fundpricing2} we have
\begin{align*}
& \cExp{\e^{-r(\tau-t)}\Ind{t\le\tau\le {t_M}}}{\Gcal_t} =\frac{f(\tau) \Ind{\tau>t}}{S_tD_t}\int_t^{t_M} \E \Big[-f(s)D_sdS_s\; \Big| \; \Fcal_t \Big]\\
&= \frac{\Ind{\tau>t}}{(a_r^\top Y_t)(a^\top Y_t)}\int_t^{t_M} f(s) \, \cExp{ -(a_r^\top Y_s)(cY_s +\gamma X_s) }{\Fcal_t}ds\\
&= \frac{\Ind{\tau>t}}{a_{Z}^\top \Ycal_t} \int_t^{t_M} f(s)\, a_{\rm D}^\top \,\e^{\Acal(s-t)} ds\begin{pmatrix}
\Ycal_t \\ \Xcal_t
\end{pmatrix}
\end{align*}
with $f(s)$ being equal to $s$ or $1$, which completes the proof.

\subsection*{Proof of Proposition~\ref{prop:lhc_zt}}

The L\'evy-Kintchine theorem shows that 
\begin{equation}\label{eq:LK}
\Psi(u)=b^Z u + \int_0^\infty (1- \e^{-u\xi})\nu^Z d\xi.
\end{equation}
We conclude the proof by applying the Sylvester's formula $\e^{UDU^{-1}}=U\e^{D}U^{-1}$ and by using~\eqref{eq:LK} in~\eqref{eq:AmatTC} as follows
\begin{align*}
\bar A &= b^Z UDU^{-1} + \int_0^\infty (\e^{UDU^{-1}\xi} - \Id)\nu^Z d\xi \\
&= b^Z UDU^{-1} + \int_0^\infty (U\e^{D\xi}U^{-1} - UU^{-1})\nu^Z d\xi \\
&= -U\Big( b^Z (-D) + \int_0^\infty (\Id - \e^{-(-D)\xi})\nu^Z d\xi\Big)U^{-1} \\
&= -U \Psi(D) U^{-1}.
\end{align*}

\subsection*{Proof of Equation~\eqref{eq:AmatGammaTC}}

The matrix $\bar A$ in Equation~\eqref{eq:AmatTC} rewrites
\begin{align*}
\bar A &= \int_0^\infty (\e^{At}-\Id)\gamma_Z t^{-1} \e^{-\lambda_Z t} dt = \gamma_Z \; \sum_{k=1}^\infty \; \frac{A^k}{k!}  \; \int_0^\infty  t^{k-1} \e^{-\lambda_Z t} dt \\
&= \gamma_Z \; \sum_{k=1}^\infty \; \frac{A^k}{k!} \frac{\Gamma(k)}{\lambda_Z^k} = \gamma_Z \; \sum_{k=1}^\infty \;\frac{\left(A \lambda_Z^{-1}\right)^k}{k} = -\gamma_Z \, \log (\Id - A \lambda_Z^{-1} )
\end{align*}
where the second equality follows from the definition of the matrix exponential, the third from the definition of the Gamma function and its values for integer values, and the last one from the definition of the matrix logarithm.

\section{Market price of risk specifications} \label{sec:MPR}

We discuss market price of risk (MPR) specifications such that $X$ has a linear drift also under the real-world measure $\Pa\sim\Q$. This may further facilitate the empirical estimation of the LHC model.

Let $\Lambda(Y_t,X_t)$ denote the time-$t$ MPR such that the drift of $X_t$ under $\Pa$ becomes
\[ \mu^\Pa_t= b  Y_t +\beta  X_t +  \Sigma(Y_t,X_t) \Lambda(Y_t,X_t) .\]
It is linear in $(Y_t,X_t)$ of the form
\[ \mu^\Pa_t=  b^\Pa Y_t + \beta^\Pa X_t  ,\]
for some vector $b^\Pa\in\R^m$ and matrix $\beta^\Pa\in\R^{m\times m}$, if and only if
\begin{equation}\label{MPRlinear}
  \Lambda_i(y,x) = \frac{ ( (b^\Pa-b) s+ (\beta^\Pa-\beta) x )_i }{\sigma_i\sqrt{x_i(y-x_i)}},\quad i=1,\dots,m.
\end{equation}

In order that $\Lambda(Y_t,X_t)$ is well defined and induces an equivalent measure change, that is, the candidate Radon--Nikodym density process
\begin{equation}\label{RNd}
  \exp\bigg( \int_0^t \Lambda(Y_u,X_u)\,dW_u - \frac{1}{2}\int_0^t \left\|\Lambda(Y_u,X_u)\right\|^2du\bigg)
\end{equation}
is a uniformly integrable $\Q$-martingale, we need that $(Y,X)$ does not attain all parts of the boundary of $E$. 
This is clarified by the following theorem, which follows from \cite{Cheridito2005}.

\begin{theorem}
The MPR $\Lambda(Y_t,X_t)$ in \eqref{MPRlinear} is well defined and induces an equivalent measure $\Pa\sim\Q$ with Radon-Nikodym density process \eqref{RNd} if, for all $i=1,\dots,m$, $X_{i0} \in (0,Y_0)$ and \eqref{eq:condatzeroS}--\eqref{eq:condatLSS} hold for the $\Q$-drift parameters $\beta,b$ and for the $\Pa$-drift parameters $\beta^\Pa,b^\Pa$ in lieu of $\beta,b$.

If, for some $i=1,\dots,m$, $\beta^\Pa_{ij}=\beta_{ij}$ for all $j\neq i$ and
\begin{enumerate}
\item $b^\Pa_i=b_i$, such that
\[ \Lambda_i(y,x) = \frac{ (\beta^\Pa_{ii}-\beta_{ii}) \sqrt{x_i}  }{\sigma_i\sqrt{y-x_i}},\]
then it is enough if $X_{i0}\in [0,Y_0)$ instead of $X_{i0}\in (0,Y_0)$ and \eqref{eq:condatzero} instead of \eqref{eq:condatzeroS} holds for $\beta_{ij},b_i$, and thus for $\beta^\Pa_{ij},b^\Pa_i$.

\item $b^\Pa_i-b_i=\beta^\Pa_{ii}-\beta_{ii}$, such that
\[ \Lambda_i(y,x) = \frac{ (\beta^\Pa_{ii}-\beta_{ii}) \sqrt{y-x_i}  }{\sigma_i\sqrt{x_i}},\]
then it is enough if $X_{i0}\in (0,Y_0]$ instead of $X_{i0}\in (0,Y_0)$ and \eqref{eq:condatLS} instead of \eqref{eq:condatLSS} holds for $\beta_{ij},b_i$, and thus for $\beta^\Pa_{ij},b^\Pa_i$.

\end{enumerate}

\end{theorem}

The assumption of linear-drift preserving change of measure is often made for parsimony and to facilitate the empirical estimation procedure. For example, the specification of MPRs that preserve the affine nature of risk-factors has been theoretically and empirically investigated in \cite{Duffee2002}, \cite{Duarte2004}, and \cite{Cheridito2007} among others.

\section{Chebyshev interpolation}\label{sec:chebint}

This Appendix describes how to perform a Chebyshev interpolation of an arbitrary function on a rectangle $[a,b]\times[c,d]\subset\R^2$.
The Chebyshev polynomials of the first kind take values in $[-1,1]$ but can be shifted and scaled so as to form a basis of $[a,b]$.
In this case they are given by the following recursion formula,
\begin{align*}
t_0^{a,b}(x) & = 1 \\
T_1^{a,b}(x) & = \frac{x-\mu}{\sigma} \\
T_{n+1}^{a,b}(x) & = \frac{2(x-\mu)}{\sigma}T_n^{a,b}(x) - T_{n-1}^{a,b}(x)
\end{align*}
with $\mu=(a+b)/2$ and $\sigma=(b-a)/2$.
The Chebyshev nodes for the interval $[a,b]$ are then given by 
\[ 
{x}^{a,b}_j = \mu + \sigma \cos\left(z_j\right), \quad z_j=\frac{(1/2+j)\pi}{N+1}, \quad \text{for $j=0, \ldots, N$.}
\]
The polynomial interpolation of order $N$ is
\[ p_N(s,x) = \sum_{n=0}^N \sum_{m=0}^N c_{n,m}\; T^{a,b}_n(s)\,T^{c,d}_m(x)\]
where the coefficients are given by
\[
c_{n,m} = 2^{\Ind{n\neq 0}+\Ind{m \neq 0}} \sum_{i=0}^N \sum_{j=0}^N \frac{f\big(x^{a,b}_i, x^{c,d}_j\big) \cos(n\,z_i) \cos(m\,z_j)}{(N+1)^2}.
\]
The coefficients can be computed in an effective way by applying Clenshaw's method, or by applying discrete cosine transform.
This straightforward interpolation has the advantage to prevent the Runge's phenomenon.
We refer to \cite{gass2018chebyshev} for more details on the multidimensional Chebyshev interpolation, and for an interesting financial application of multivariate function interpolation in the context of fast model estimation or calibration.

\end{appendix}


\begin{table}
\begin{center}
\subfloat[Bombardier Inc.]{
\begin{tabular}{rrrrrrrrr}
  \hline
 & all & 1 yr & 2 yrs & 3 yrs & 4 yrs & 5 yrs & 7 yrs & 10 yrs \\ 
  \hline
Mean & 274.51 & 144.07 & 194.80 & 243.38 & 279.43 & 329.40 & 357.10 & 373.71 \\ 
  Vol & 165.23 & 156.66 & 158.95 & 153.31 & 147.95 & 141.14 & 130.46 & 121.64 \\ 
  Median & 244.76 & 94.79 & 145.71 & 189.55 & 232.44 & 295.51 & 353.01 & 376.58 \\ 
  Min & 28.02 & 28.02 & 39.22 & 59.50 & 86.64 & 109.58 & 146.32 & 171.29 \\ 
  Max & 1288.71 & 1288.71 & 1151.92 & 1092.74 & 1062.57 & 1048.33 & 960.16 & 887.06 \\ 
   \hline
\end{tabular}

} \\
\subfloat[Walt Disney Co.]{
\begin{tabular}{rrrrrrrrr}
  \hline
 & all & 1 yr & 2 yrs & 3 yrs & 4 yrs & 5 yrs & 7 yrs & 10 yrs \\ 
  \hline
Mean & 31.01 & 11.97 & 17.53 & 22.74 & 28.90 & 34.59 & 45.00 & 56.18 \\ 
  Vol & 21.85 & 12.93 & 15.73 & 17.18 & 18.18 & 18.15 & 16.13 & 15.66 \\ 
  Median & 26.30 & 7.70 & 12.42 & 17.39 & 24.31 & 30.45 & 42.98 & 55.58 \\ 
  Min & 1.63 & 1.63 & 3.24 & 4.47 & 5.81 & 8.18 & 12.92 & 17.51 \\ 
  Max & 133.02 & 79.38 & 102.20 & 115.19 & 120.62 & 126.43 & 127.22 & 133.02 \\ 
   \hline
\end{tabular}

}
\end{center}
\caption{CDS spreads summary statistics.}
The sample contains 552 weekly observations collected between January 1\textsuperscript{st} 2005 and January 1\textsuperscript{st} 2015 summing up to 3620 CDS spreads in basis point for each firm.\label{tab:cdsstats}
\end{table}

\begin{table}
\begin{center}
\subfloat[Bombardier Inc.]{
\begin{tabular}{rrrr}
  \hline
 & LHCC(2) & LHCC(3) & LHCC(3)$^*$ \\ 
  \hline
$\gamma_1$ & 0.205 & 0.201 & \textbf{0.400} \\ 
  $\kappa_1$ & 0.546 & 1.263 & 1.316 \\ 
  $\kappa_2$ & 0.421 & 0.668 & 0.884 \\ 
  $\kappa_3$ &  & 0.385 & 0.668 \\ 
  $\theta_1$ & 0.624 & 0.841 & 0.696 \\ 
  $\theta_2$ & 0.512 & 0.699 & 0.548 \\ 
  $\theta_3$ &  & 0.478 & 0.401 \\ 
   \hline
\end{tabular}

} \\
\subfloat[Walt Disney Co.]{
\begin{tabular}{rrrr}
  \hline
 & LHCC(2) & LHCC(3) & LHCC(3)$^*$ \\ 
  \hline
$\gamma_1$ & 0.056 & 0.064 & \textbf{0.130} \\ 
  $\kappa_1$ & 0.167 & 0.258 & 0.294 \\ 
  $\kappa_2$ & 0.165 & 0.229 & 0.280 \\ 
  $\kappa_3$ &  & 0.091 & 0.212 \\ 
  $\theta_1$ & 0.666 & 0.753 & 0.558 \\ 
  $\theta_2$ & 0.662 & 0.721 & 0.536 \\ 
  $\theta_3$ &  & 0.298 & 0.387 \\ 
   \hline
\end{tabular}

}
\end{center}
\caption{Fitted and fixed (in bold) parameters for the LHC models.}
\label{tab:paramsLCRM}
\end{table}

\begin{table}
\begin{center}
\subfloat[Bombardier Inc.]{
\begin{tabular}{rrrrrrrrrr}
  \hline
  & & all & 1 yr & 2 yrs & 3 yrs & 4 yrs & 5 yrs & 7 yrs & 10 yrs \\ 
  \hline
\multirow{4}{*}{${\rm LHCC}(2)$}
& RMSE & 26.24 & 23.87 & 31.79 & 24.13 & 12.31 & 24.36 & 27.70 & 33.33 \\ 
& Median & -0.22 & -13.90 & -3.16 & -1.23 & 4.63 & 20.20 & -0.17 & -18.90 \\ 
& Min & -83.96 & -64.23 & -83.96 & -65.09 & -22.09 & -20.50 & -38.64 & -79.80 \\ 
& Max & 123.86 & 123.86 & 43.98 & 32.90 & 39.31 & 57.07 & 75.58 & 54.45
\\[10pt]
\multirow{4}{*}{${\rm LHCC}(3)$} 
& RMSE & 16.10 & 8.90 & 19.63 & 19.46 & 11.01 & 17.35 & 15.93 & 16.94 \\ 
& Median & -0.25 & 1.14 & -7.69 & -5.47 & 1.06 & 16.46 & 2.06 & -9.42 \\ 
& Min & -56.64 & -24.62 & -56.64 & -52.93 & -31.01 & -0.66 & -12.85 & -46.56 \\ 
& Max & 107.23 & 107.23 & 23.86 & 15.42 & 20.38 & 41.61 & 49.57 & 31.94
 \\[10pt]
\multirow{4}{*}{${\rm LHCC}(3)^*$}
& RMSE & 21.87 & 9.07 & 23.52 & 24.01 & 12.67 & 16.56 & 25.15 & 32.37 \\ 
& Median & -0.42 & 0.02 & -4.22 & -3.94 & -3.12 & 14.22 & -0.66 & -4.80 \\ 
& Min & -82.13 & -24.32 & -66.96 & -68.24 & -32.91 & -31.95 & -54.44 & -82.13 \\ 
& Max & 67.51 & 24.43 & 25.10 & 26.16 & 22.24 & 42.51 & 67.51 & 59.33  \\
   \hline
\end{tabular}
} \\
\subfloat[Walt Disney Co.]{
\begin{tabular}{rrrrrrrrrr}
  \hline
  & & all & 1 yr & 2 yrs & 3 yrs & 4 yrs & 5 yrs & 7 yrs & 10 yrs \\ 
  \hline
\multirow{4}{*}{${\rm LHCC}(2)$}
& RMSE & 2.88 & 3.09 & 1.66 & 2.73 & 2.82 & 2.82 & 2.00 & 4.30 \\ 
& Median & -0.33 & -0.13 & -0.86 & -1.99 & -1.40 & -0.43 & 1.40 & 1.10 \\ 
&  Min & -12.65 & -12.65 & -4.15 & -5.21 & -4.34 & -4.32 & -5.54 & -12.64 \\ 
&  Max & 8.81 & 3.58 & 5.11 & 8.81 & 8.70 & 8.22 & 4.62 & 6.43
\\[10pt]
\multirow{4}{*}{${\rm LHCC}(3)$} 
& RMSE & 1.06 & 0.85 & 1.09 & 1.02 & 0.89 & 1.31 & 1.33 & 0.75 \\ 
& Median & -0.03 & 0.35 & 0.19 & -0.55 & -0.43 & 0.14 & 0.70 & -0.26 \\ 
& Min & -5.57 & -4.87 & -5.57 & -3.53 & -3.55 & -4.34 & -4.62 & -1.97 \\ 
& Max & 4.94 & 2.74 & 4.94 & 3.58 & 4.34 & 3.85 & 3.53 & 2.68
 \\[10pt]
\multirow{4}{*}{${\rm LHCC}(3)^*$}
& RMSE & 1.17 & 1.02 & 1.11 & 0.98 & 1.15 & 1.62 & 1.07 & 1.12 \\ 
& Median & 0.01 & 0.47 & 0.35 & -0.62 & -0.60 & -0.06 & 0.48 & -0.02 \\ 
& Min & -5.48 & -5.45 & -5.48 & -3.49 & -3.78 & -4.83 & -3.92 & -4.65 \\ 
& Max & 4.63 & 2.68 & 4.49 & 3.28 & 4.63 & 3.98 & 2.98 & 4.15 \\
   \hline
\end{tabular}
}
\end{center}
\caption{Comparison of CDS spreads fits for the LHC models.}
The tables report the minimal, maximal, median, and root mean squared errors in basis point by maturity over the entire time period for the three different specifications.
\label{tab:rmse}
\end{table}


\begin{figure}
\begin{center}
\includegraphics{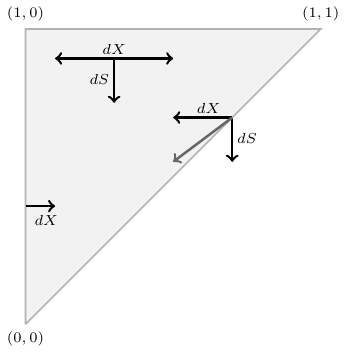}
\end{center}
\caption{State space of the LHC model with a single factor.}
Illustrations of the inward pointing drift conditions~\eqref{eq:condatzero}--\eqref{eq:condatLS}. The survival process value is given by the $y$-axis and the factor value by the $x$-axis.\label{fig:1Fss}
\end{figure}

\begin{figure}
\begin{center}
\includegraphics{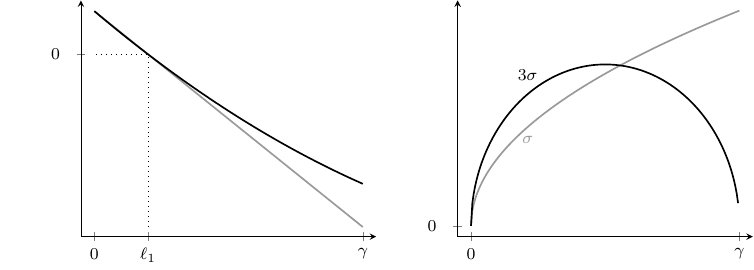}
\end{center}
\caption{Comparison of the one-factor LHC and CIR models.}
Drift and diffusion functions of the default intensity for the one-factor LHC model (black line) and affine model (grey line).
The parameter values are $\ell_1=0.05$, $\ell_2=1$, and $\gamma=0.25$. \label{fig:1Fdynamic}
\end{figure}

\begin{figure}
\begin{center}
\includegraphics{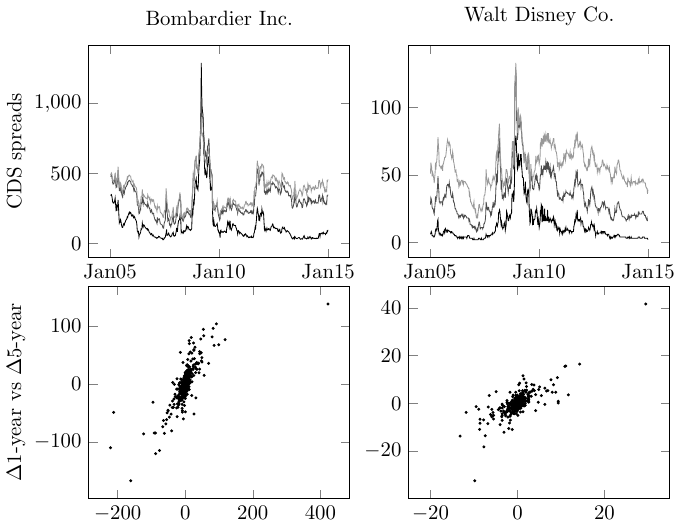}
\end{center}
\caption{CDS spread data.}
Panels on the first row display the CDS spreads in basis points for the maturities 1 year (black), 5 years (grey), and 10 years (light-grey).
Panels on the second row display the weekly changes in 1-year versus 5-year CDS spreads.\label{fig:cdsts}
\end{figure}

\begin{figure}
\begin{center}
\subfloat[Bombardier Inc.]{
\includegraphics{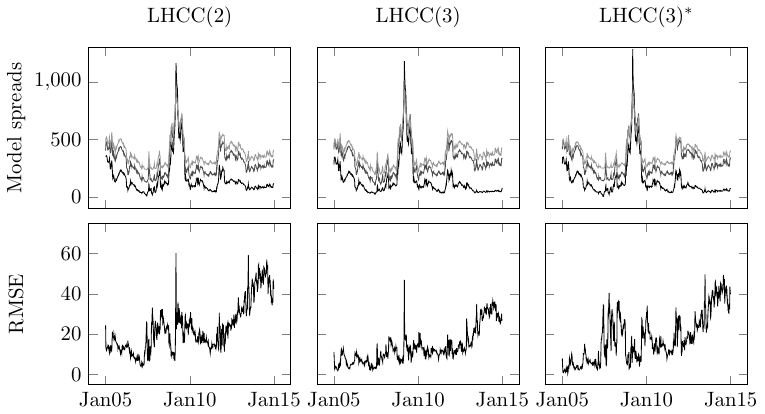}
} \\
\subfloat[Walt Disney Co.]{
\includegraphics{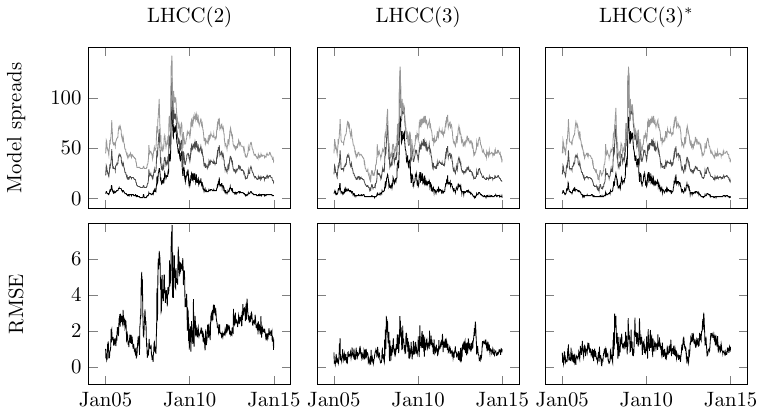}
}
\end{center}
\caption{CDS spreads fits and errors.}
Panels on the first row display the fitted CDS spreads in basis points with maturities 1 year (black), 5 years (grey), and 10 years (light-grey) for the three specifications.
Panels on the second row display the root-mean-square error (in basis points) computed every day and aggregated over all the maturities.\label{fig:fitLCRM}
\end{figure}

\begin{figure}
\begin{center}
\subfloat[Bombardier Inc.]{
\includegraphics[scale=0.95]{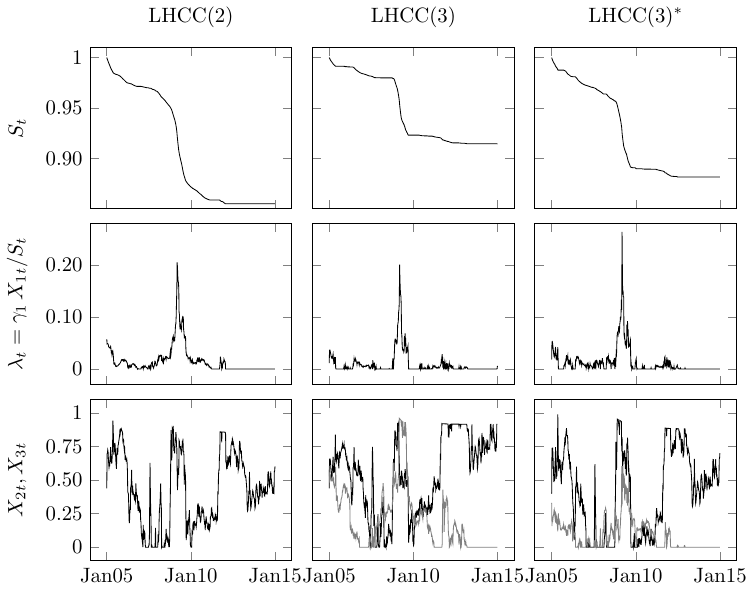}
} \\
\subfloat[Walt Disney Co.]{
\includegraphics[scale=0.95]{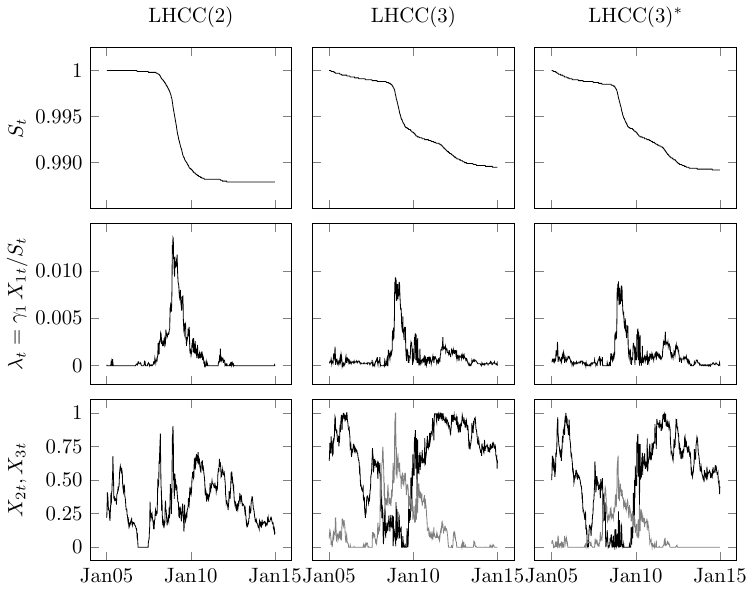}
}
\end{center}
\caption{Factors fitted from CDS spreads.}
The filtered factors of the three estimated specifications are displayed over time.
Panels on the first row display the drift only survival process, panels o the second row the implied default intensity, and panels on the last row the process $X_{m}$ in black and the process ${X_{2}}$ in grey for the three-factor models.
\label{fig:factors}
\end{figure}

\begin{figure}
\begin{center}
\includegraphics{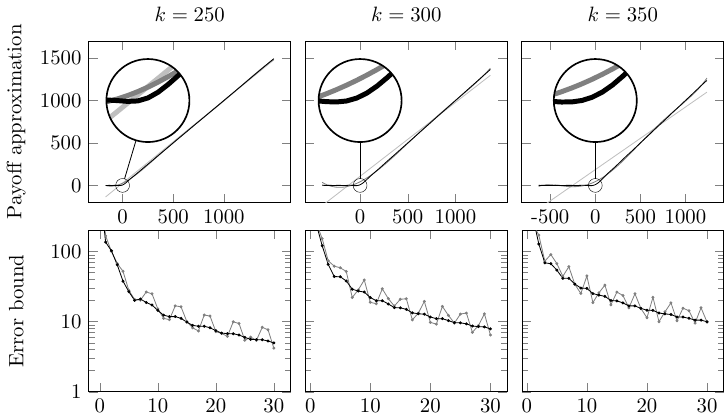}
\end{center}
\caption{CDS option payoff approximations.}
Panels on the first row display the polynomial interpolation of the payoff function approximation with the Fourier-Legendre approach at the order 1 (light-grey), 5 (grey), and 30 (black).
Panels on the second display the price error bound with the Fourier-Legendre approach (black) and with the Chebyshev approach (grey) as functions of the polynomial interpolation order.
The first (second and third) column corresponds to a CDS option with a strike spread of 250 (300 and 350) basis points.
All values are reported in basis points.
\label{fig:CDSO_payoff}
\end{figure}

\begin{figure}
\begin{center}
\includegraphics{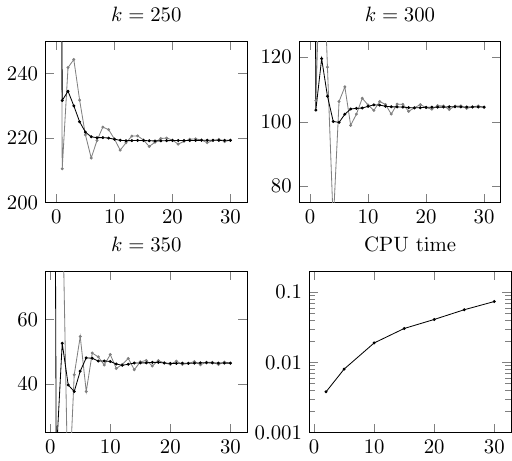}
\end{center}
\caption{CDS option price approximations and CPU times.}
The top and bottom-left panels display the price approximation with the Fourier-Legendre approach (black) and with the Chebyshev approach (grey) as functions of the polynomial interpolation order.
The top-left (top-right and bottom-left) panels corresponds to a CDS option with a strike spread of 250 (300 and 350) basis points. All values are reported in basis points.
The bottom-right panel displays the CPU times in seconds needed to compute the price approximation as functions of the polynomial interpolation order.
\label{fig:CDSO_pricecpu}
\end{figure}

\begin{figure}
\begin{center}
\includegraphics{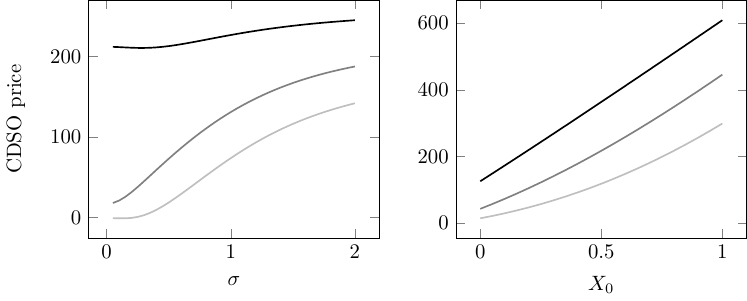}
\end{center}
\caption{CDS option price sensitivities.} 
The figure on the left (on the right) displays the CDS option price as a function of the volatility parameter (the initial risk factor position) for the strike spread 250 (black), 300 (grey), and 350 (light-grey).
All values are reported in basis points.
\label{fig:CDSoptSens}
\end{figure}

\processdelayedfloats

\newpage
\bibliographystyle{chicago}
\bibliography{lcrm}

\end{document}